\documentclass[12pt, draftclsnofoot, onecolumn]{IEEEtran}
\usepackage{pgfplots}
\pgfplotsset{compat=newest}
\usepackage{epsfig}
\usepackage{amsthm,amssymb,graphicx,graphicx,multirow,amsmath,color,epstopdf,subfloat}
\usepackage{caption}
\usetikzlibrary{plotmarks}
\usetikzlibrary{patterns}
\usetikzlibrary{calc,positioning,fit,backgrounds}
\usetikzlibrary{shapes,snakes}
\usetikzlibrary{intersections,positioning}
\usepackage{subcaption}
\usepackage[noadjust]{cite}
\include{notation}

\newtheorem{lem}{Lemma}

\newcommand{\SIR}{\mathtt{SIR}}
\newcommand{\SINR}{\mathtt{SINR}}

\newcommand{\p}[2]{$P_{#1}^{#2}$}
\newcommand{\PL}[2]{\|{#1}\|^{#2}}
\newcommand{\al}{\alpha}

\newcommand{\E}{\mathbb{E}}
\def\D{\mathrm{d}}
\newcommand\blfootnote[1]{%
  \begingroup
  \renewcommand\thefootnote{}\footnote{#1}%
  \addtocounter{footnote}{-1}%
  \endgroup
}

\usepackage{hyperref}

\begin{document}
\bibliographystyle{ieeetr}
%
\title{Joint Backhaul-Access Analysis of Full Duplex Self-Backhauling Heterogeneous Networks}
\author{\IEEEauthorblockN{Ankit Sharma, Radha Krishna Ganti and  J. Klutto Milleth}}

\maketitle

\begin{abstract}

With the successful demonstration of in-band full-duplex (IBFD) transceivers, a new research dimension has been added to wireless networks.
This paper proposes a use case of this capability for IBFD self-backhauling 
heterogeneous networks (HetNet). IBFD self-backhauling in a HetNet refers to IBFD-enabled small cells backhauling themselves 
with macro cells over the wireless channel. Owing to their IBFD capability, the small cells simultaneously communicate over the 
access and backhaul links, using the same frequency band. The idea is doubly advantageous, as it obviates the need for fiber 
backhauling small cells every hundred meters and allows the access spectrum to be reused for backhauling at no extra cost.
This work considers the case of a two-tier cellular network with IBFD-enabled small cells, 
wirelessly backhauling themselves with conventional macro cells. For clear exposition, the case considered is that of FDD network, 
where within access and backhaul links, the downlink (DL) and uplink (UL) are frequency duplexed ($f1$, $f2$ respectively), while 
the total frequency spectrum used at access and backhaul ($f1+f2$) is the same. Analytical expressions for 
coverage and average downlink (DL) rate in such a network are derived using tools from the field of \emph{stochastic geometry}. 
It is shown that DL rate in such networks could be close to double that of a conventional TDD/FDD self-backhauling network, at the 
expense of reduced coverage due to higher interference in IBFD networks. For the proposed IBFD network, the conflicting aspects of 
increased interference on one side and high spectral efficiency on the other are captured into a mathematical model. The mathematical model
introduces an end-to-end joint analysis of backhaul (or fronthaul) and access links, in contrast to the largely available access-centric studies. \noindent \blfootnote{\noindent Ankit Sharma and Radha Krishna Ganti were with the Electrical Engineering Department at Indian Institute of 
Technology, Madras at the time of submission. Their contact emails are \{ankitsharma, rganti\}@ee.iitm.ac.in. Ankit Sharma is now with Cypress Semiconductors India Pvt. Ltd. (ansa@cypress.com). J. Klutto Milleth is with Center of Excellence in Wireless Technology, IIT-Madras Research Park. His contact email is klutto@cewit.org.in}.
\end{abstract}
\IEEEpeerreviewmaketitle
\section{Introduction}
Capacity demands in a wireless cellular system have been increasing at a rapid pace. The next move towards 5G network aims at increasing capacity 
of the current systems thousand fold \cite{Metis}. Since bandwidth demands have ever been exceeding the available spectrum, 
frequency reuse techniques are becoming increasingly important for cellular systems. The well studied dense heterogeneous network (HetNet) 
\cite{3gppHetNets} is one of the methods to increase capacity for future networks. Typically, HetNet consists of a macro base-station (M-BS) 
tier, serving high mobility users overlaid with operator deployed pico base-station (P-BS) tier (a.k.a. small cells) \cite{3gpp.36.932} for 
low mobility, dense user areas. Deploying a highly dense network of P-BSs is becoming increasingly worrisome \cite{Wireless2020} for 
operators. This is because fiber backhauling such P-BSs placed every few tens of meters is not a practically and economically 
viable option, especially in developing countries like India. The alternative is to employ wireless backhauling. Though wireless 
backhauling obviates the need for laying down high-speed/fiber links, it needs the operator to partition their highly priced spectrum into 
orthogonal access and backhauling resources, thereby resulting in lower spectral usage for user access.
\par
In-band full-duplex (IBFD) systems---another frequency reuse technique---present a scheme to wirelessly backhaul P-BSs with M-BSs without having to orthogonalize allocated spectrum between access and backhaul. The scheme consists of a two-tier cellular network where the P-BSs, being IBFD-enabled, backhaul themselves wirelessly with the M-BSs, which themselves are fiber-backhauled to the core network. The M-BSs exchange backhaul data with the P-BSs on the entire spectrum that the P-BSs use to transmit data to the users. M-BSs may also serve the users directly. 
Since practical IBFD radio systems (\cite{nadh2016linearization}, \cite{Arjun}, \cite{Jain:2011:PRF:2030613.2030647} and \cite{ASabharwal}) have already been demonstrated, the proposed scheme results in an amalgamation of two frequency reuse techniques working in tandem. 
To this end, the paper analyzes and gives key design insights 
for a future cellular network\footnote{Since the paper studies a two-tier HetNet architecture based on each tier being FDD in its own uplink and downlink, comparison of IBFD-enabled networks will be done with the conventional FDD systems (with no IBFD-enabled station) throughout the paper.} that leverages the efficiency of IBFD radios used in a wirelessly backhauled two-tier HetNet.

\subsection{Related Work}
For a self-backhauled two-tier HetNet, a model for joint analysis of backhaul-access links is required, which is a rather less studied topic. 
The topic finds mention in \cite{Katti}, where it is listed as one of the potential applications of IBFD radios. Work in \cite{BoyuLi} is 
an attempt in this direction, though the work develops on the basic assumption of one P-BS per user and inter P-BS interference 
has not been considered. Moreover, the paper only presents capacity results as a function of physical 
separation between the P-BS and M-BS while the overall coverage trends in such a two-tier network have not been analyzed. 
Perhaps a closely related work in this direction is found in \cite{Marios}, where the authors model a 
multiple-input-multiple-output (MIMO) IBFD P-BS and conventional half-duplex (HD) backhauling M-BS. The M-BSs only play the role of backhaul aggregators and do not provide access communication to users. The probability of successful transmissions is modeled 
as a product of independent successful transmissions for first hop (M-BS to P-BS) and second hop (P-BS to user) in the downlink (DL). This might not be always true of real systems where there might be dependence between the two probabilities. Also, the aggregate rate characterization from the M-BS to the user 
has not been detailed.
\par
\vspace{3px}
Other works like \cite{GoyalS} analyze an IBFD network for parameters like rate but only for a single-tier network. 
They allocate same channels to both uplink (UL) and DL of base station-to-user link and compute the 
parameters thereof. Work in \cite{Hyungjong} discusses the optimal power allocation strategy in IBFD networks using relays. 
The work builds on a cognitive setup with primary and secondary nodes in general. Interference is then controlled from primary 
transmitters to secondary receivers. The approach is modeled as an optimization problem for transmit powers of primary 
and secondary transmitters. Works in \cite{Barghi} and \cite{Hyungsik} discuss about bringing in 
MIMO and beamforming on IBFD radios and the benefits thereof, though \cite{Barghi} uses only a single tier. 
Two interesting analyses are offered through works in \cite{AggarwalV} and \cite{Riihonen} where the authors argue the use of IBFD at all. 
The authors pitch the use case of using multiple antennas for the conventional HD multiple-input-multiple-output (MIMO) 
operation versus using the antennas for IBFD operation. In fact, most of the cases discuss only the access link optimization. Another 
relevant study in self-backhauling is the recent work in \cite{singh2014tractable}. The authors present the system level coverage and 
rate results in a mesh network of base-stations (BS) with wired backhaul, providing wireless backhauling for BSs without wired backhaul. 
However, the study is done for millimeter-wave networks without IBFD capability. Previous work on similar HetNet architecture was presented in \cite{Ankit}, but was limited to a single path loss exponent being used for both the P-BS as well as the M-BS tier. This work generalizes \cite{Ankit} to two different path loss exponents which is practically more relevant.

\subsection{Our approach and novelty}
The paper proposes a two-tier network consisting of IBFD-enabled P-BSs and conventional M-BSs. It analyzes the performance of the sytem in 
the DL. The setup consists of P-BSs being wirelessly backhauled by the M-BSs. Since the P-BSs 
are IBFD-enabled, they use the same set of frequencies to backhaul themselves on the DL and UL with the M-BS, as the ones they 
use in the DL and UL access links to the users (say, $f1$ and $f2$ be the DL and UL frequency for the P-BS to user 
(and M-BS to P-BS in backhaul) link and user to P-BS (and P-BS to M-BS in backhaul) link transmissions). 
The M-BSs being conventional non-IBFD stations, need to bifurcate frequency resources between backhaul and access links. For $1$ Hz of 
bandwidth, the M-BSs use $\eta$ Hz ($0 \le \eta \le 1$) for backhauling and $(1-\eta)$ Hz for direct access links to users.
It is interesting to note that the design fits \emph{as-is} for a frequency division duplexed IBFD network and could be 
tailored to suit other networks, such as TDD as well. Moreover, the design requires only the P-BSs to be IBFD, while the user 
devices and M-BS could work on legacy FDD mode (refer Fig.~\ref{fig:Fig1_DLIntf_full-duplex}).\\
For the given two-tier network, Poisson Point Process (PPP) (\cite{Stoyan} and \cite{ganti2012stochastic}) is used for the 
spatial distribution of nodes (P-BS and M-BS). The main contributions of this work are listed below:
\begin{itemize}
	\item A novel HetNet architecture, leveraging IBFD capability is proposed and the coverage probability and average rate for a typical user in such a network are derived. 
	\item The paper achieves mathematical derivation of the exact coverage and rate parameters for the proposed IBFD HetNet. Though it is intuitive to see that spectrum reuse increases rates at the expense of decreased coverage due to wireless backhaul links, an exact quantification of these two contrasting effects has been established in this work. Tractable and quickly computable coverage expressions are important for system analysis of future IBFD-enabled HetNets. The analysis also identifies inter-tier 
	interference and the bandwidth division at the backhauling M-BS as the main limiting factors in such HetNets.
	\item In the proposed network, the effective signal-to-interference ($\SIR$) ratio distribution for a typical user associated with a P-BS is modeled as the joint $\SIR$ distribution of the $\{$user--P-BS, P-BS--M-BS$\}$ link-pair.  Therefore the coverage under P-BS implies joint coverage -- of the typical user under a P-BS, along with coverage of the same P-BS with a backhauling M-BS. The average rate for a P-BS associated user is modeled as the minimum of rates on the $\{$user--P-BS, P-BS--M-BS$\}$ link-pair. This introduces inter-dependence between the two tiers.
		\end{itemize}
In \cite{Harpreet}, the coverage probability was obtained in 
a general K-tier HetNet, but without any dependence between the tiers themselves. 
In the proposed network, since the backhaul links are also active over the wireless channel, interference to access links of users is enhanced and the coverage degrades. On the other hand, reusing the access spectrum for wireless backhauling in an IBFD setting tends to double the spectral efficiency of the system. This work models and details the way 
these two contrasting factors affect the overall system behavior.
\section{System Model}
The system model considered in this paper is described in the following sub-sections.
\subsection{Spatial arrangement of base-stations}
The location of the M-BSs and P-BSs are assumed to follow independent Poisson point processes $\Phi_m \subset \mathbb{R}^2$ and 
$\Phi_s\subset \mathbb{R}^2$ with densities $\lambda_m$ and $\lambda_s$, respectively. 
The transmit powers of the M-BS and P-BS tier are assumed to be $P_m$ and $P_s$ respectively. Small scale fading between any pair of nodes 
is assumed to be independent and Rayleigh distributed. The fading power (square of the small scale fading) between nodes located at points 
$x$ and $y$ in $\mathbb{R}^2$ is denoted by $g_{xy}$ and is exponentially distributed, also with unit mean. 
Basic large scale path loss function is used, i.e., the  power received at distance $r$ when transmitting at unity power is given as 
$r^{-\alpha}$, where $\alpha > 2$ is the path loss exponent. Path loss exponents for M-BS and P-BS tiers are denoted by $\al_m$ and $\al_s$, respectively.
Without loss of generality, a typical user located at the origin is considered and the performance of this typical user in the 
DL is analyzed.

\subsection{Association Model}
\label{subsec:Subsec_AssocModel}
The association rule is based on the maximum average received biased power as discussed in \cite{HanShin}. Biasing a user to 
associate with a P-BS even if the received power from a M-BS is higher, helps offload traffic from the M-BSs. Hence, for BS association, 
the average received biased power at the typical user is $P_sB_s \PL{x_s}{-\alpha_s}$ and $P_mB_m \PL{x_m}{-\alpha_m}$ for P-BS and M-BS respectively, 
where $B_s$ and $B_m$, and $x_s$ and $x_m$ 
represent their respective biases and distances from the typical user at the origin. Let $x_{s,min}$ and $x_{m,min}$ denote the distance of 
the closest P-BS and M-BS respectively, to the user at the origin. Then the user connects to the P-BS if $ x_{m,min} \geq  \Delta_m^{-1} x_{s,min}^{\al_s / \al_m}$ 
and to the closest M-BS, otherwise. Here $\Delta_m =( (P_s B_s)/(P_m B_m))^{1/\al_m}$. 
Let $\varepsilon_m$ and $\varepsilon_s$ denote the events of M-BS and P-BS association respectively, of the typical user.
Then the corresponding probabilities of association are given in \cite{HanShin} as,
\begin{equation}
\Pr(\varepsilon_s) = 2\pi \lambda_s \int_{0}^{\infty} e^{-\pi\left(\lambda_m \Delta_m^{-2} x_{s,min}^\frac{2\al_s}{\al_m} + \lambda_s x_{s,min}^2\right)} x_{s,min}\, \D x_{s,min}\; ;\,\,\,\,\,\,\,\,\,\,\,\, \Pr(\varepsilon_m) = 1 - \Pr(\varepsilon_s).
\label{eqn:Eqn1_PrAssoc}
\end{equation}

\subsection{Bandwidth Allocation}
\label{subsec:Subsec1_BWAlloc}
Bandwidth allocation between the P-BS and M-BS tiers is discussed next, considering $2W$ Hz of allocated spectrum. 
\subsubsection*{Full-Duplex Bandwidth Allocation}
 For IBFD networks, the available spectrum of $2W$ Hz is allocated as: 
 \begin{enumerate}
  \item The entire $2W$ Hz is used by P-BSs and M-BSs.
  \item Within each tier, $2W$ Hz is divided into UL and DL resources utilizing $W$ Hz each (as for conventional FDD). 
  \item At the M-BSs (being non-IBFD), $W$ Hz is further sub-divided as $\eta W$ Hz and $(1-\eta)W$ Hz, $0 \leq \eta \le 1$, 
 for backhaul and access resources respectively.
 \item Also, each M-BS to P-BS link is limited in bandwidth to $\left(\frac{\eta}{n}\right)W$ Hz, considering each M-BS backhauls 
 $n = \lambda_s/\lambda_m$ P-BSs on an average.
 \end{enumerate} 
 \subsubsection*{Half-Duplex Bandwidth Allocation}
 For conventional FDD networks, the available $2W$ Hz is allocated as: 
 \begin{enumerate}  
 \item $\kappa\,2W$ Hz and $(1-\kappa)(2W)$ Hz, $0 \leq \kappa < 1$, partitioned between M-BSs and P-BSs respectively. Typically, 
 $\kappa = 0.5$, so each tier gets $W$ Hz. Notice that this is in contrast to both the tiers getting the entire $2W$ Hz in IBFD case.
 \item At each tier, $W$ Hz is divided into UL and DL resources utilizing $W/2$ Hz each.  
 \item At the M-BSs, $W/2$ Hz is further sub-divided as $\eta W/2$ Hz and $(1-\eta)W/2$ Hz, $0 \leq \eta \le 1$, 
 for backhauling and access resources respectively.
 \item Also, each M-BS to P-BS link is limited in bandwidth to $\left(\frac{\eta}{n}\right)W/2$ Hz, considering each M-BS backhauls 
 $n = \lambda_s/\lambda_m$ P-BSs on an average.
 \end{enumerate}
 \begin{subfigures}
 \begin{figure}[!ht]
 \centering
 \includegraphics[scale = 0.25]{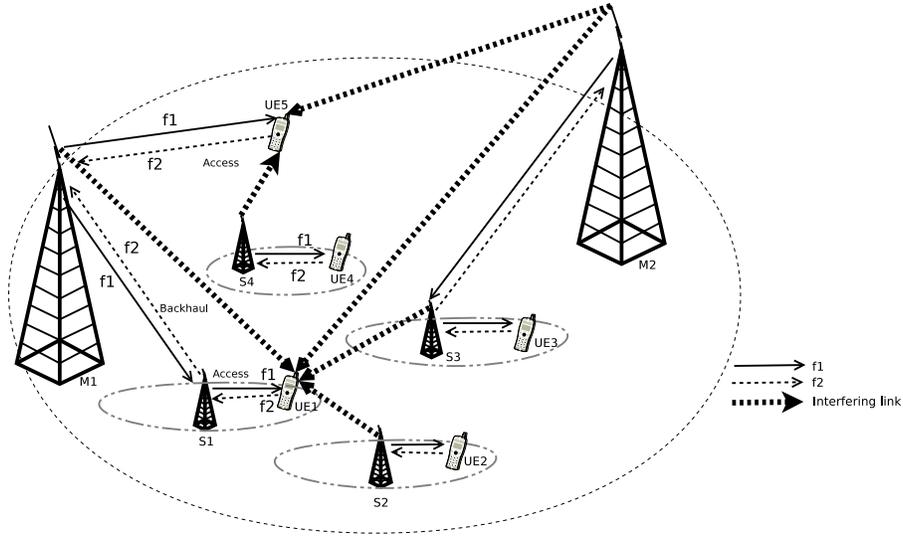}
 \captionsetup{font=footnotesize}
 \caption{DL interference in IBFD system. Total spectrum = $2W$ Hz. Each link represents a bandwidth of $W$ Hz. For instance, the DL backhaul link is 
 centered around $f1$ Hz and has a bandwidth of $W$ Hz. Users attached to either P-BS or M-BS receive interference from 
 both the tiers. The given spectrum though, is entirely used by both the tiers.}
 \label{fig:Fig1_DLIntf_full-duplex}
\end{figure}
 \begin{figure}[ht!]
 \centering
 \includegraphics[scale = 0.25]{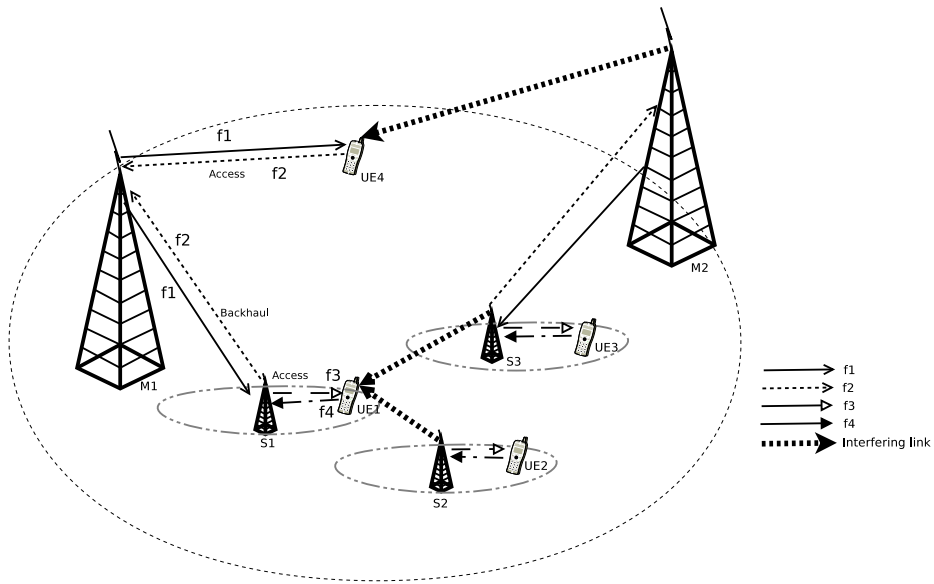}
 \captionsetup{font=footnotesize}
 \caption{DL interference in conventional FDD system. Total spectrum = $2W$ Hz. Each link represents a bandwidth of $W/2$ Hz. For instance, the DL backhaul link is 
 centered around $f1$ Hz and has a bandwidth of $W/2$ Hz. Users attached to a tier (P-BS or M-BS) receive interference from 
 only from that tier. However, the given spectrum needs to be partitioned between the tiers.}
 \label{fig:Fig2_DLIntf_half-duplex}
\end{figure}
\end{subfigures}
Taking the case of an IBFD system, the proposed frequency allocation plan is depicted in Fig.~\ref{fig:Fig1_DLIntf_full-duplex}. 
For conventional FDD system the frequency plan is well known and depicted in Fig.~\ref{fig:Fig2_DLIntf_half-duplex}. The figures denote DL 
and UL carriers as $f1$ and $f2$ respectively, that are centered about the bandwidth of $W$ Hz and $W/2$ Hz in IBFD and conventional FDD 
case respectively. In IBFD systems, both P-BSs and M-BSs use the total available spectrum but interference is more, as shown by the thick 
broken lines in Fig.~\ref{fig:Fig1_DLIntf_full-duplex}. For conventional FDD systems, though the interferers are only the nodes belonging 
to the tier to which the user is associated, the total available spectrum is partitioned between the M-BS and P-BS.
\subsection{Signal-to-Interference Ratio}
An interference limited network is assumed and signal-to-interference-plus-noise ratio ($\SINR$) is replaced 
$\SIR$ \cite{Harpreet} as the measure of performance.
\subsubsection{Small cell association}
Consider a typical user at the origin associated with a P-BS. Let point $r_s\in \Phi_s$ denote this closest P-BS to the typical user. 
Let the point $r_m \in \Phi_m$ denote the closest M-BS to the aforementioned P-BS. The P-BS associates with the closest M-BS for backhaul.  Let $\SIR_{us}$ denote $\SIR$ of the signal from the P-BS to the user in DL access. Then 
\begin{align}
\SIR_{us}(r_s,r_m) = \frac{P_sg_{o r_s}\|r_s\|^{-\alpha_s}}{I_s(o)+I_m(o)+P_m g_{o r_m}\|r_m\|^{-\alpha_m}  }
\label{eqn:Eqn2_SIRus}
\end{align}
where,
\[I_s(x) = \sum_{z\in\Phi_s\cap B(o,r_s)^c}P_sg_{x z}\|z-x\|^{-\alpha_s}, \]
is the interference from other P-BSs to a user located at a point $x$ in $\mathbb{R}^2$ and $B(o,r_s)$ denotes a disc centered at origin $o$, having radius $r_s$ and  $B(o,r_s)^c$ denotes its complement. 
The interference from the M-BS to a user located at a point $x$ in $\mathbb{R}^2$ is 
\[I_m(x) = \sum_{z\in\Phi_m\cap \mathcal{M}}P_mg_{x z}\|z-x\|^{-\alpha_m},\]
where $ \mathcal{M} =( B(o,r_s^{\al_s / \al_m}\Delta_m^{-1})\cup B(r_s,\|r_m - r_s\|))^c$, and the discs are assumed to be open sets.  The $\SIR$ of the signal from the M-BS to the P-BS in DL backhaul is then 
given as,
\begin{align}
\SIR_{sm}(r_s,r_m) =  \frac{P_mg_{r_s r_m}\|r_m-r_s\|^{-\alpha_m}}{I_s(r_s)+I_m(r_s)+\beta P_m  },
\label{eqn:Eqn3_SIRsm}
\end{align}
where the residual self-interference generated by the P-BS, being IBFD, is modeled as $\beta P_m$, 
$\beta$ being a factor controlling the amount of self-interference. Though 
the self-interference channel in some of the recent literature (\cite{Duarte,Marios}) has been modeled as a Rician fading channel 
\cite{simon2005digital}, this paper focuses on a simpler model. The idea is to get a handle on network coverage and rates given 
a self-interference suppressing IBFD radio, than to quantify the self-interference suppression capability of an IBFD radio.
\subsubsection{Macro cell association} Assume that the typical user at the origin is associated to an M-BS denoted by point at $r_m' \in 
\Phi_m$. Let $\mathtt{SIR}_{um}$ denote the $\SIR$ of the signal from the M-BS to the user in DL access and is given by
\begin{align}
\mathtt{SIR}_{um}(r_m') =  \frac{P_mg_{or_m'}\|r_m'\|^{-\alpha_m}}{\hat{I}_s(o) +\hat{I}_m(o)  }. 
\label{eqn:Eqn4_SIRum}
\end{align}
where $\hat{I}_s(o)= \sum_{z\in\Phi_s\cap B(o,\Delta_m r_m'^{\al_m / \al_s})^c}P_sg_{o z}\|z\|^{-\alpha_s}$ and 
$\hat{I}_m(o)=  \sum_{z\in\Phi_m\cap B(o,r_m')^c}P_mg_{o z}\|z\|^{-\alpha_m}$.
The next section analyzes coverage probability of a typical user in the given network.
\section{Coverage}
Coverage  probability is defined as the probability that a randomly chosen user in the given network achieves an $\SIR$ greater than a given threshold. Let $T_s,\,T_b \text{ and } T_m$ be the $\SIR$ coverage thresholds for user to P-BS, P-BS to M-BS and user to M-BS links respectively.
In the proposed setup, the effective coverage for a P-BS associated user would depend jointly on user to P-BS and P-BS to M-BS coverage probabilities denoted as \p{u,s}{}($T_s$), 
\p{s,m}{}($T_b$). For an M-BS associated user, coverage would only depend on the user to M-BS coverage probability denoted 
as \p{u,m}{}($T_m$). Using \eqref{eqn:Eqn1_PrAssoc}, the effective coverage probability for a user, \p{u}{x}($T_s,T_b,T_m$), can now be defined as
\begin{equation}
 P_u^{x}(T_s,T_b,T_m) = \Pr(\varepsilon_s) \cdot \Pr\left(\SIR_{us} > T_s,\, \SIR_{sm} > T_b \mid \varepsilon_s\right) + \Pr(\varepsilon_m) \cdot \Pr\left(\SIR_{um} > T_m \mid \varepsilon_m\right),
 \label{eqn:Eqn5_Pcu}
\end{equation}
where $\varepsilon_m \text{ and } \varepsilon_s$ denote events of M-BS and P-BS association and $x \in \{f,h\}$ denoting IBFD (full-duplex) or 
conventional FDD (half-duplex) operation.  
The joint distribution of  $r_m$ and $r_s$, that will be used in the evaluation of coverage probability is discussed next.

\subsection{Joint probability density function of distance pair $(r_s,\,r_m)$}
\label{subsec:Subsec_JtPdf}
As mentioned above, the coverage under P-BS association implies a joint coverage probability over user to P-BS and P-BS to its 
backhauling M-BS links. This entails deriving a joint probability density function (pdf) of the distance pair $(r_s,\,r_m)$ 
with respect to a typical user at the origin. When the user associates with a P-BS, the joint pdf $f(r_s,r_m)$ is derived for a  general $\Delta_m$, that is, $\Delta_m \ge 1$ (typical, P-BS biased association) and $0 < \Delta_m < 1$ (negative P-BS bias). In Fig.~\ref{fig:Fig4_PcSc_SysArrangements} and Fig.~\ref{fig:Fig5_PcSc_SysArrangements_DeltaLess1}, the possible spatial configurations of the user, P-BS and M-BS are shown that occur because of various possible relative locations of the user associated P-BS and P-BS associated M-BS, with respect to the typical user at the origin.
\begin{subfigures}
\begin{figure}[!ht]
\centering
 \includegraphics[scale = 0.27]{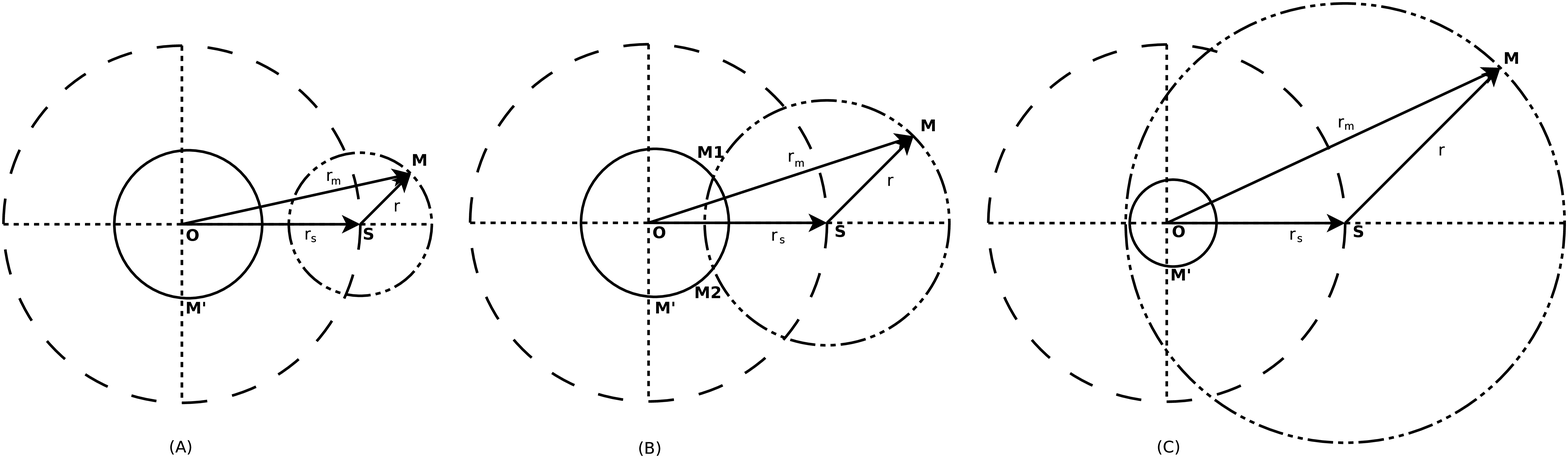}
 \captionsetup{font=footnotesize}
 \caption{Network geometry for event $\varepsilon_s$ With $\Delta_m \ge 1$, i.e., user biased towards P-BS tier. The three possible scenarios are as a result of different spatial locations of the P-BS and M-BS with respect to the typical user at the origin.  When the user associates with a P-BS, coverage depends  jointly on user to P-BS (for access) link and P-BS to M-BS (for backhaul) link. 
 Given a P-BS $S$, found at distance $r_s$ from $O$, the nearest M-BS to the user could be at a distance $\Delta_m^{-1}r_s^{\al_s / \al_m}$ from 
 the origin $O$, denoted by $OM'$. 
 The backhauling M-BS $M$ could be found anywhere at a distance $r_m$ from $O$, resulting in three different network geometries ($(A)\,,(B)\text{ and } (C)$) that define the joint density function 
 of the P-BS and M-BS with respect to the typical user.}
 \label{fig:Fig4_PcSc_SysArrangements}
\end{figure}
Instead of deriving the joint distribution for $(r_s,r_m)$, an equivalent distribution of $(r_s, r)$ is derived. This is because of the occurrence of the term $\PL{r_m-r_s}{-\al}$ in the $\SIR_{sm}(r_s,r_m)$ expression of equation~\ref{eqn:Eqn3_SIRsm}. Replacing it with an equivalent $\PL{r}{-\al}$ simplifies the derivation of coverage expressions and so the joint distribution on $(r_s,r)$ is used.
\begin{lem}
 The joint density function of the access-backhaul distance pair, $(r_s, r)$, with respect to the typical user, given the bias factor $\Delta_m \geq 1$ is 
 
 \begin{equation}
    f(r_s,r)=
    \begin{cases}
    \vspace{10px}
      \displaystyle\frac{4 e^{-\pi\left(r^2\lambda_m + \frac{r_s^{\frac{2\al_s}{\al_m}}}{\Delta_m^2}\lambda_m + r_s^2 \lambda_s \right)} \pi^2 r \lambda_m \left( \frac{r_s^{\frac{2\al_s}{\al_m}}}{\Delta_m^2} \al_s \lambda_m + r_s^2 \al_m \lambda_s\right)} {r_s \al_m}, & 0 < \|r\| \le \nu_-(r_s,\Delta_m,\al_s,\al_m) \\
      \vspace{10px}
     \displaystyle\frac{\partial\left( e^{-\lambda_s \pi r_s^2}\, e^{-\lambda_m \left(\pi(\Delta_m^{-1}r_s^{\al_s / \al_m})^{2} + \pi r^2 - \,\operatorname{lens}(M_1,\, M_2) \right)}\right)}{\partial r_s \partial r}, & \|r\| \in \nu_-^+(r_s,\Delta_m,\al_s,\al_m)\\
   4 \pi ^2 \lambda_m \lambda_s r\, r_s e^{-\pi  \left(\lambda_m r^2+\lambda_s r_s^2\right)}, & \|r\| \ge \nu_+(r_s,\Delta_m,\al_s,\al_m), 
    \end{cases}
    \label{eqn:Eqn_pdf_DeltaGrThan1}
  \end{equation}
  where \begin{itemize}
  	\item $ \nu_-(r_s,\Delta_m,\al_s,\al_m) \triangleq \, \|r_s\| - \Delta_m^{-1}\|r_s\|^{\al_s / \al_m}$
  	\item $ \nu_+(r_s,\Delta_m,\al_s,\al_m) \triangleq \, \|r_s\| + \Delta_m^{-1}\|r_s\|^{\al_s / \al_m}$
  	\item $ \nu_-^+(r_s,\Delta_m,\al_s,\al_m) \triangleq \, \displaystyle\left] \|r_s\| - \Delta_m^{-1}\|r_s\|^{\al_s / \al_m},\, \|r_s\| + \Delta_m^{-1}\|r_s\|^{\al_s / \al_m} \right]$
  	\item $\operatorname{lens}(M_1,M_2)$ denotes the area of the lens formed between the points $M_1$ and $M_2$ in Case (B) of 
  	Fig.~\ref{fig:Fig4_PcSc_SysArrangements}
  \end{itemize}
\end{lem}

\begin{proof}
See Appendix~\ref{app:App1_AB_PDF}
\end{proof}
\begin{figure}[!ht]
\centering
 \includegraphics[scale = 0.27]{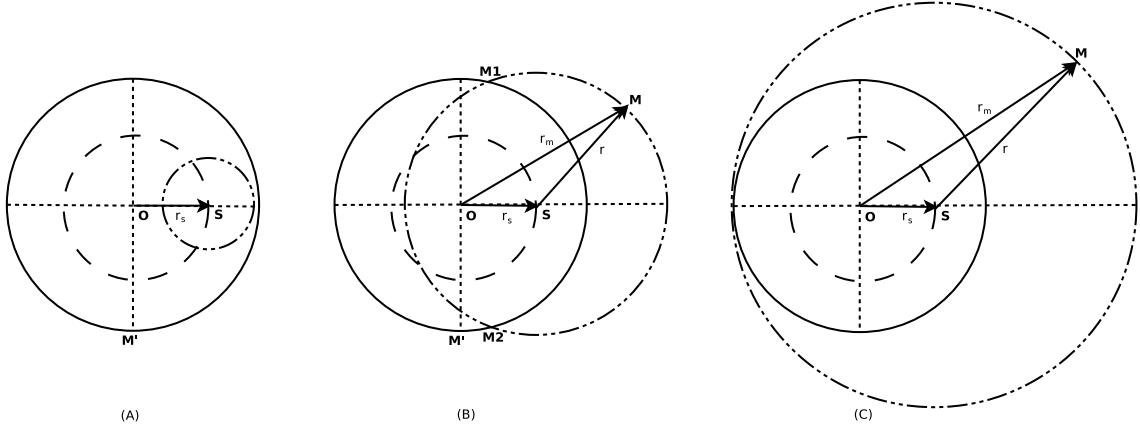}
 \captionsetup{font=footnotesize}
 \caption{Network Geometry for Event $\varepsilon_s$ With $0 < \Delta_m < 1$. 
 The figure is similar to Fig.~\ref{fig:Fig4_PcSc_SysArrangements}, except that the user is biased towards the M-BS tier. 
 Given the user associated P-BS is at a distance $r_s$ from $O$, the nearest M-BS could only be at a distance 
 $\ge$ $\Delta_m^{-1}r_s^{\al_s / \al_m}$.}
 \label{fig:Fig5_PcSc_SysArrangements_DeltaLess1}
\end{figure}
\end{subfigures}
\begin{lem}
 The joint density function of the access-backhaul distance pair, $(r_s, r)$, with respect to the typical user, given the bias factor $\Delta_m < 1$ is 
 
 \begin{equation}
    f(r_s,r)=
    \begin{cases}
    \vspace{10px}
    0, & 0 < \|r\| \le \mu_-(r_s,\Delta_m,\al_s,\al_m) \\
    \vspace{10px}
    \displaystyle\frac{\partial\left( e^{-\lambda_s \pi r_s^2}\, e^{-\lambda_m \left(\pi(\Delta_m^{-1}r_s^{\al_s / \al_m})^{2} + \pi r^2 - \,\operatorname{lens}(M_1,\, M_2) \right)}\right)}{\partial r_s \partial r}, & \|r\| \in \mu_-^+(r_s,\Delta_m,\al_s,\al_m)\\
    4 \pi ^2 \lambda_m \lambda_s r\, r_s e^{-\pi  \left(\lambda_m r^2+\lambda_s r_s^2\right)}, & \|r\| \ge \mu_+(r_s,\Delta_m,\al_s,\al_m), 
    \end{cases}
    \label{eqn:Eqn_pdf_DeltaLeThan1}
  \end{equation}
  where \begin{itemize}
  	\item $ \mu_-(r_s,\Delta_m,\al_s,\al_m) \triangleq \, -\|r_s\| + \Delta_m^{-1}\|r_s\|^{\al_s / \al_m}$
  	\item $ \mu_+(r_s,\Delta_m,\al_s,\al_m) \triangleq \, \|r_s\| + \Delta_m^{-1}\|r_s\|^{\al_s / \al_m}$
  	\item $ \mu_-^+(r_s,\Delta_m,\al_s,\al_m) \triangleq \, \displaystyle\left] -\|r_s\| + \Delta_m^{-1}\|r_s\|^{\al_s / \al_m},\, \|r_s\| + \Delta_m^{-1}\|r_s\|^{\al_s / \al_m} \right]$
  	\item $\operatorname{lens}(M_1,M_2)$ denotes the area of the lens formed between the points $M_1$ and $M_2$ in Case (B) of 
  	Fig.~\ref{fig:Fig5_PcSc_SysArrangements_DeltaLess1}
  \end{itemize}
\end{lem}

\begin{proof}
See Appendix~\ref{app:App1_AB_PDF}
\end{proof}
\subsection{Small Cell Coverage in IBFD}
\label{subsec:Subsec_SC_Covg}
In this section, the coverage probability of a typical user under P-BS is derived. Coverage under P-BS is denoted as 
\p{u,s}{f}($T_s,T_b$) and the corresponding geometry of the node locations is depicted in Fig.~\ref{fig:Fig3_PcSc}.

\begin{figure}[ht!]
\centering
 \includegraphics[width = 6cm]{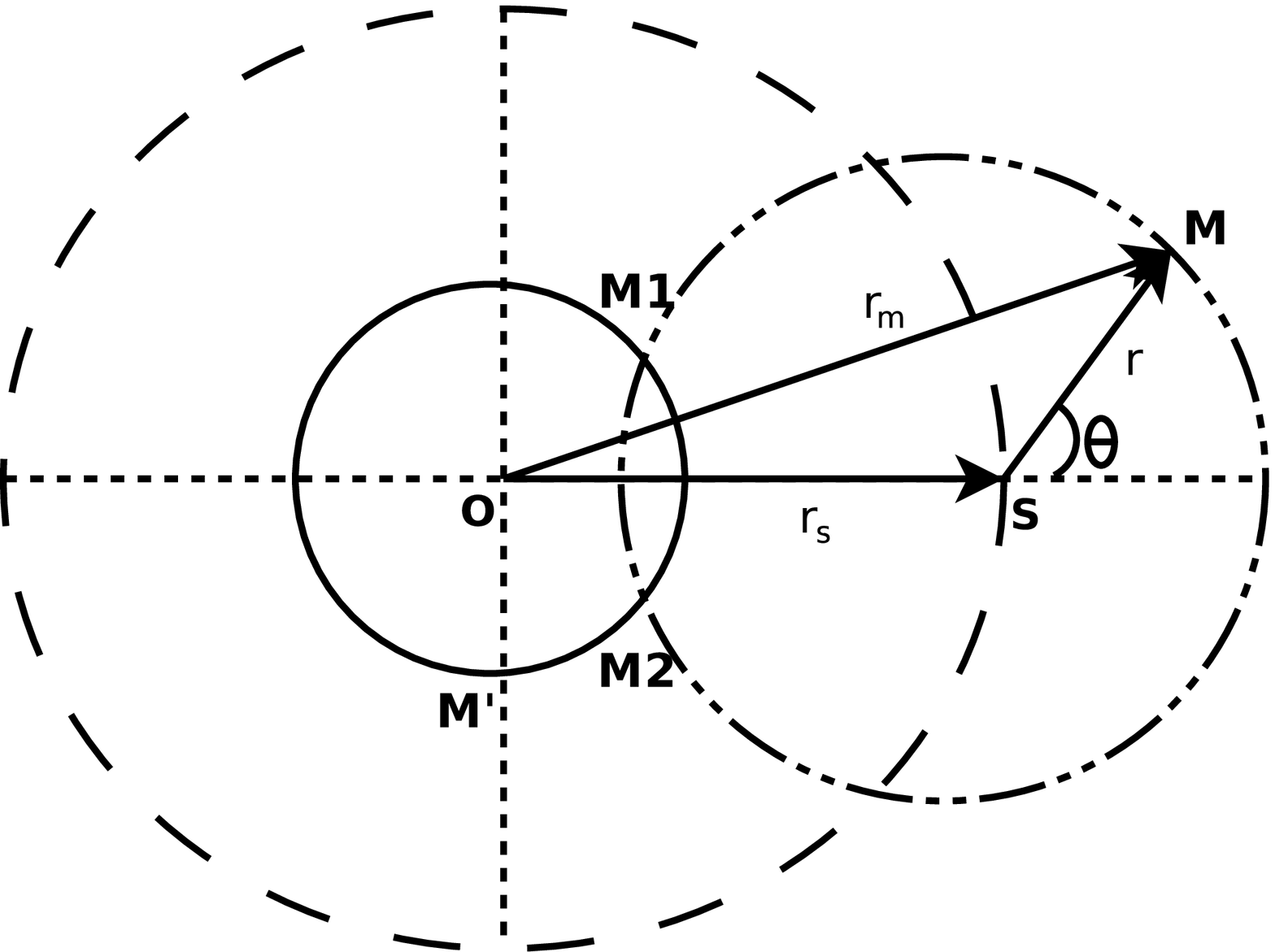}
 \captionsetup{justification=centering}
 \caption{An instance of user associating with P-BS}
 \label{fig:Fig3_PcSc}
\end{figure}
A typical user located at the origin $o$ associates with a P-BS ($S$ in Fig.~\ref{fig:Fig3_PcSc}) at a distance $r_s$. 
From \eqref{eqn:Eqn1_PrAssoc}, it follows that there is no M-BS inside a ball of radius $OM' = \Delta_m^{-1} r_s^{\al_s / \al_m}$ centered at the origin $o$. 
For the backhaul, the P-BS $S$ connects to the nearest M-BS ($M$ in Fig.~\ref{fig:Fig3_PcSc}) at a distance $r_m$ from $o$. The backhaul 
distance from the P-BS $S$ to the backhauling M-BS $M$ is $r$. In IBFD mode, the user when associated with the P-BS, will receive 
interference from other P-BSs as well as all the M-BSs.
Let  $g(s,x,\alpha)$ be defined as
\[ g(s,x,\alpha) = \frac{1}{1+s\PL{x}{-\alpha}} .\]

\begin{lem}
The probability of coverage for a user associated with a P-BS in the given two-tier IBFD network is 
\begin{equation}
\begin{aligned}
	&P_{u,s}^{f}(T_s,T_b) =\\
&\int\limits_0^{2 \pi}\int\limits_{r_s>0, r>0} e^{-\lambda_s \int\limits_{\Phi_s \cap A_1^c} 1 - g(s_1, z, \alpha_s) g(s_2, z-r_s, \alpha_s)\D{z} -\lambda_m \int\limits_{\Phi_m \cap A_2^c} 1 - g(s_1^{\prime}, v, \alpha_m) g(s_2^{\prime}, v-r_s, \alpha_m)\D{v} -\beta s_2} g(s_1^{\prime}, r_m, \alpha_m)f(r_s,r) \D{r_s}\D{r}\, \D{\theta},
\label{eqn:Eqn6_Pcus_FD}
\end{aligned}
\end{equation}
where, $A_1 = B(o,r_s)$, $s_1 = T_s\|r_s\|^{\alpha_s}$, $s_2 = \frac{T_b}{P_m}\|r\|^{\alpha_m}P_s$, $r_m = \sqrt{r_s^2 + r^2 + 2 r_s r \cos{\theta}}$,  $A_2 = ( B(o,\Delta_m^{-1} r_s^{\al_s / \al_m})\cup B(r_s,\|r\|))$, $s_1^{\prime} = \frac{T_s}{P_s}\|r_s\|^{\alpha_s}P_m$, and   $s_2^{\prime} = T_b \|r\|^{\alpha_m}$.
\label{lem:Lemma1}
\end{lem}

\begin{proof}
 See Appendix~\ref{app:App2_PcSc}
\end{proof}

The integral in Lemma \ref{lem:Lemma1} can be  dived  into three integrals over the variables $\theta, r_s$ and $r$ corresponding to the cases \emph{  $(A),\,(B) \text{ or } (C)$} of Fig.~\ref{fig:Fig4_PcSc_SysArrangements} 
and Fig.~\ref{fig:Fig5_PcSc_SysArrangements_DeltaLess1}. Of particular interest, is the density function of 
\emph{Case $(B)$}, where the backhaul and the inner macro discs intersect. Backhaul disc is the one that has distance from the serving P-BS 
to the serving P-BS's backhauling M-BS as the radius. Though the expression for it has been derived as in \eqref{eqn:Eqn_pdf_DeltaGrThan1}, 
it is hard to compute numerically. Therefore, probability for the intersection case is analyzed below.\\
Let $\mathcal{C}$ and $\mathcal{I}$ denote events \emph{user covered under P-BS} and \emph{intersection of the 
backhaul and the inner macro discs} of Fig.~\ref{fig:Fig4_PcSc_SysArrangements} respectively. Then $\mathcal{I}$ is defined as
\[\mathcal{I} \triangleq \|r-\Delta_m^{-1} r_s^{\al_s / \al_m}\|~\le~\|r_s\|~\le~\|r+\Delta_m^{-1} r_s^{\al_s / \al_m}\|.\] 
Using Bayes rule, the probability $\Pr\left({\mathcal{I} \mid \mathcal{C}}\right)$  is
\begin{equation}
\begin{aligned}
\Pr\left({\mathcal{I} \mid \mathcal{C}}\right)\,\,= \frac{\Pr\left({\mathcal{C},\,\mathcal{I}}\right)}{\Pr\left({\mathcal{C}}\right)}.
\end{aligned}
\end{equation}
 The expressions for 
$\Pr\left( \mathcal{C},\mathcal{I} \right)$ and $\Pr\left( \mathcal{C}\right)$ are already derived in equation \eqref{eqn:Eqn6_Pcus_FD}.

 \begin{figure}[ht!]
 \centering 
 \captionsetup{font=footnotesize}
 \begin{tikzpicture}
 \begin{axis}[scale=1.2,
 grid = both,
 legend cell align=left,
 legend style={align=left,at={(1.0,0.5)},anchor=east},
 xlabel ={$B_s$ (dB, P-BS Bias)},
 ylabel = {Probability of Network Topology }
 ]
 
\addplot[only marks,color=black]
coordinates{
(22.000000, 0.999869) (24.000000, 0.999876) (26.000000, 0.999883) (28.000000, 0.999887) (30.000000, 0.999890) (32.000000, 0.999893) (34.000000, 0.999896) (36.000000, 0.999897) (38.000000, 0.999897) (40.000000, 0.999898) (42.000000, 0.999898) (44.000000, 0.999899) (46.000000, 0.999899) (48.000000, 0.999900) (50.000000, 0.999900) (52.000000, 0.999900) (54.000000, 0.999900) (56.000000, 0.999900) (58.000000, 0.999900) (60.000000, 0.999900) 
};\addlegendentry{Total Prob.}

\addplot+[mark=triangle, color=blue]
coordinates{
(22.000000, 0.037584) (24.000000, 0.052359) (26.000000, 0.069831) (28.000000, 0.084436) (30.000000, 0.099195) (32.000000, 0.114060) (34.000000, 0.126002) (36.000000, 0.136367) (38.000000, 0.147434) (40.000000, 0.155487) (42.000000, 0.161879) (44.000000, 0.170067) (46.000000, 0.173075) (48.000000, 0.176709) (50.000000, 0.179773) (52.000000, 0.183063) (54.000000, 0.184969) (56.000000, 0.188636) (58.000000, 0.192738) (60.000000, 0.192539) 
};\addlegendentry{Case (A)}

\addplot+[mark=oplus, color=red]
coordinates{
(22.000000, 0.283656) (24.000000, 0.257451) (26.000000, 0.231230) (28.000000, 0.205454) (30.000000, 0.179870) (32.000000, 0.156807) (34.000000, 0.136777) (36.000000, 0.119342) (38.000000, 0.101966) (40.000000, 0.086671) (42.000000, 0.073626) (44.000000, 0.062513) (46.000000, 0.052814) (48.000000, 0.044764) (50.000000, 0.038781) (52.000000, 0.033469) (54.000000, 0.029132) (56.000000, 0.024488) (58.000000, 0.020871) (60.000000, 0.015603)   
};\addlegendentry{Case (B)}

\addplot+[mark=o, color=green]
coordinates{
(22.000000, 0.678649) (24.000000, 0.690082) (26.000000, 0.698833) (28.000000, 0.710005) (30.000000, 0.720832) (32.000000, 0.729031) (34.000000, 0.737119) (36.000000, 0.744190) (38.000000, 0.750498) (40.000000, 0.757741) (42.000000, 0.764394) (44.000000, 0.767320) (46.000000, 0.774011) (48.000000, 0.778427) (50.000000, 0.781346) (52.000000, 0.783368) (54.000000, 0.785799) (56.000000, 0.786776) (58.000000, 0.786290) (60.000000, 0.791758)  
};\addlegendentry{Case (C)}
\end{axis}
 \end{tikzpicture}
 \caption{Probability of Network Topology vs. P-BS Bias in Small cell association ($T_m = T_s = T_b = -10\,$dB. $B_m = 0\,$dB, $P_m = 22\,$dB, $P_s = 0$ dB, $\al_m = 2.8,\, \al_s = 4,\,\lambda_s = 4 \lambda_m$). Each plot shows the probability of the network being in a particular geometry. Case (A) denotes \emph{Zero Intersect. Prob.} curve that 
 depicts the probability of the inner macro and the backhaul disc having zero intersection, Case (B) (\emph{Intersect Prob.}) curve 
 depicts the probability of the inner macro and the backhaul disc intersecting and Case (C) (\emph{Engulf. Prob.})
 curve depicts the probability of the backhaul disc engulfing the inner macro disc
 in event of small cell association of Fig.~\ref{fig:Fig4_PcSc_SysArrangements}. Notice that in the limit of bias towards P-BS, 
 i.e. high bias towards small cell tier, the \emph{Intersect Prob.} curve  goes to $0$, rendering numerical computations 
 much easier.}
 \label{fig:Fig_PcSc_SystemState}
 \end{figure}
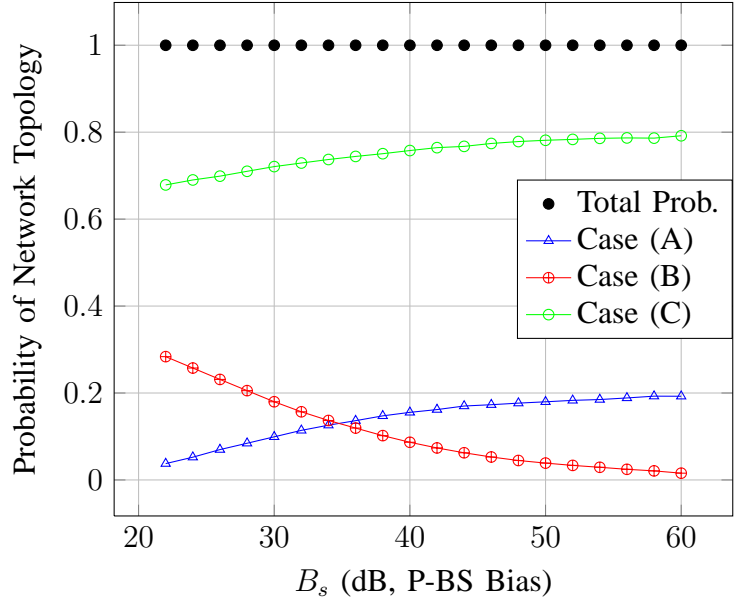
The plot in Fig.~\ref{fig:Fig_PcSc_SystemState} reveals useful information about the network topology. For the given system model, 
at reasonably high biasing towards the P-BS, the system mostly remains in the state of \emph{Case $(A)$} or \emph{Case $(C)$} of 
Fig.~\ref{fig:Fig4_PcSc_SysArrangements}. Therefore, coverage could be approximated by averaging over system states of \emph{Case $(A)$} and  \emph{Case $(C)$} alone, which is much more tractable than using the entire joint density function---a rather complex function to evaluate. The plot also makes practical sense, as a HetNet under typical circumstances, would be operated in a mode highly biased towards the P-BSs \cite{damnjanovic2011survey,Airspan1}.

\subsection{Macro Cell Coverage in IBFD}
\label{subsec:Subsec_MC_Covg}
Coverage probability for a user associated with an M-BS is derived here. For such a user, there is only a single active link (user-M-BS) 
since the M-BSs are fiber backhauled to the core network. In this case, it is more convenient to calculate coverage as 
$\Pr\left(\SIR_{um} > T_m, \varepsilon_m\right)$ directly rather than the conditional coverage based on the event $\varepsilon_m$.
\begin{lem}
The probability of coverage for a user associated with a M-BS in the given two-tier IBFD network is denoted by \p{u,m}{f}($T_m$) and is given as
\begin{equation}
\begin{aligned}
P_{u,m}^{f}(T_m) &= \int\limits_{r_m' = 0}^\infty \int\limits_{r_s = \Delta_s r_m'^{\al_m / \al_s}}^\infty F(\Phi_m, \Phi_s) f(r_m',\, r_s) \, \D r_m' \D r_s,
\end{aligned}
\label{eqn:Eqn7}
\end{equation}
where $f(r_m',\, r_s)$ denotes the density function of the nearest M-BS and P-BS and is given as:
\begin{equation}
f(r_m',\, r_s) = 2 \pi \lambda_m r_m' e^{-\pi \lambda_m r_m'^2} 2 \pi \lambda_s r_s e^{-\pi \lambda_s r_s^2}
 \label{eqn:Eqn8}
 \end{equation} and
\[ F(\Phi_m,\Phi_s) =  e^{-\pi r_m'^2 \lambda_m T_m^{2/\alpha_m} \int\limits_{T_m^{-2/\alpha_m}}^\infty \frac{1}{1 + t^{\alpha_m/2}}\, \D t}\;\;e^{ -\pi r_m'^{2 \al_m / \al_s} \lambda_s\left(\frac{P_s T_m}{P_m}\right)^{\frac{2}{\alpha_s}} \int\limits_{\left(\frac{B_s}{B_m T_m}\right)^{\frac{2}{\alpha_s}} }^\infty \frac{1}{1+t^{\alpha_s/2}}\,\D t  }. \]
\label{lem:Lemma2}
\end{lem}
\begin{proof}
See Appendix~\ref{app:App_MBS_Covg_IBFD}
\end{proof}

\subsection{Small Cell Coverage in FDD}
In FDD case the frequency resources are orthogonalized between the access and backhaul tiers and so interference to a user in the DL is much reduced. This 
comes at the cost of halving the spectrum for access and backhaul link each.

\begin{lem}
 The probability of coverage for a user associated with a P-BS in the given two-tier FDD network is denoted by \p{u,s}{h}($T_s,T_b$) and given as
\begin{multline}
 P_{u,s}^{h}(T_s,T_b) = \int\limits_{\mathbb{R}^2} e^{ -\lambda_s \int\limits_{z \in \Phi_s \cap A_1^c} 1 - g\left(\PL{r_s}{\alpha_s}T_s,z,\alpha_s\right) \, \D z -\lambda_m \int\limits_{v \in \Phi_m \cap A_2^c}  1 - g\left(\PL{r}{\alpha_m}T_b,v-r_s,\alpha_m\right) \, \D v}\:
 f(r_s,r) \, \D r_s \, \D r,
 \label{eqn:Eqn11_PcSc_FDD}
\end{multline}
\text{where } $g(s,x,\alpha) = \frac{1}{1+s\PL{x}{-\alpha}},\,\,\,\, A_1 = B(o,r_s),$  $A_2 = ( B(o,\Delta_m^{-1} r_s^{\al_s / \al_m})\cup B(r_s,\|r\|))$.
\label{lem:Lemma3}
\end{lem}

\begin{proof}
\begin{equation*}
 \Pr\left(\SIR_{us} > T_s,\, \SIR_{sm} > T_b \mid \varepsilon_s\right) = \E_{r_s,r}\left[ \underbrace{\Pr\left(\SIR_{us} > T_s,\, \SIR_{sm} > T_b \mid r_s,\,r\right)}_{P_{u,s}^{'h}(T_s,T_b)} \right].
\end{equation*}
The representation of conditioning on points $r_s$ and $r$ is dropped in interest of better clarity, 
for the following derivation. For FDD case, the user (or P-BS) receives interference only from the tier that it is associated 
with.

\begin{equation*}
 P_{u,s}^{'h}(T_s,T_b) = \Pr\left(  \frac{P_s g_{o r_s} \PL{r_s}{-\al_s}}{\sum\limits_{z \in A_1^c}P_s g_{o z} \PL{z}{-\al_s} } > T_s, \right. \\
   \left. \frac{P_m g_{r_s r_m} \PL{r}{-\al_m}}{ \sum\limits_{z \in A_2^c}P_m g_{r_s z} \PL{z-r_s}{-\al_m}} > T_b   \right), \\
\end{equation*}

Proceeding in the same way as in Appendix \ref{app:App2_PcSc} for IBFD coverage, the expression for coverage in the FDD network 
could be calculated to be as \eqref{eqn:Eqn11_PcSc_FDD}.
\end{proof}

\subsection{Macro Cell Coverage in FDD}
Users associated with the macro cells in FDD case see interference only from the macro cells. 
The coverage expression uses $r_m'$ and $r_s$ for nearest P-BS and M-BS respectively as in IBFD macro cell coverage case. So 
macro cell coverage is calculated as $\Pr(\SIR_{um} > T_m, \varepsilon_m)$ directly.

\begin{lem}
\label{lem:Lemma4}
The probability of coverage for a user associated with a M-BS in the given two-tier FDD network is denoted by \p{u,m}{h}($T_m$) and given as,

\begin{equation}
P_{u,m}^{h}(T_m) = \int\limits_{r_m'=0}^\infty \int\limits_{r_s = \Delta_s r_m'^{\al_m / \al_s}}^\infty e^{-\pi r_m'^2 \lambda_m T_m^{2/\alpha_m} \int\limits_{T_m^{-2/\alpha_m}}^\infty \frac{1}{1+t^{\alpha_m/2}}\, \D t} f(r_m',r_s)\, \D r_m' \, \D r_s,
\end{equation}

where $f(r_m',r_s)$ is defined as in \eqref{eqn:Eqn8}.
\end{lem}
\begin{proof}
Lemma~\ref{lem:Lemma4} directly follows from the proof given for Lemma~\ref{lem:Lemma2}, considering a user associated with a 
given tier will receive interference only from that tier.
\end{proof}
\section{Average Rate}
This section focuses on the the achievable rate for a typical user located at the origin conditioned on the user being under 
coverage. For full-duplex
case the entire $1\,Hz$ is used for self-backhauling as well as access links by the P-BSs. At the M-BSs, $\eta\,Hz$ is used for the backhauling link 
to P-BSs and an orthogonal $(1-\eta)\,Hz$ for direct access link to the user. The arrangement is similar for half-duplex case, but for the fact that the 
spectrum is orthogonalized as $0.5\,Hz$ each, for access and backhaul links with respect to the P-BSs. Notice that the rate in DL for users connected to the P-BS 
is the minimum of rates on the M-BS to P-BS and the P-BS to user links. This is taken into account by the derivation that follows.
Let an event, that the user is covered, be defined as $\{\text{Coverage}\}\, \triangleq \, \mathbf{1}(\varepsilon_m)\{\SIR_{um} > T_m\}\,\, \cup \,\,\mathbf{1}(\varepsilon_s)\{\SIR_{us}>T_s,\,\SIR_{sm}>T_b\}$, 
where $\mathbf{1}(\varepsilon)$ denotes an indicator random variable for event $\varepsilon$,
\begin{equation}
\begin{aligned}
 \E\left[R_u \mid \text{Coverage}\right] &=\,\, \frac{1}{\Pr\{\text{Coverage}\}}\displaystyle\left( \E\left[R_{um}\mid \SIR_{um} > T_m\right] \Pr(\SIR_{um} > T_m)\, + \right. \\ 
 &\left. \E\left[{\min}\left(R_{us},\, R_{sm}\right)\mid \SIR_{us} > T_s,\, \SIR_{sm} > T_b\right] \Pr(\SIR_{us} > T_s,\, \SIR_{sm} > T_b).\displaystyle\right)
\end{aligned}
\label{eqn:Eqn_Rate}
\end{equation}
\subsection{M-BS to User Rate ($R_{um}$)}
The expectation in the first term in \eqref{eqn:Eqn_Rate} is calculated as follows. Let $\overline{\eta} = \left(1-\eta \right)$. Let $\overline{\eta}\,Hz$ be used at the M-BSs for access to user. Then, 
\begin{equation}
\E\left[R_{um} \mid \SIR_{um} > T_m\right] =\frac{\overline{\eta}}{\Pr(\SIR_{um} > T_m)}\int\limits_{t > 0}  \Pr(\SIR_{um}> \max(2^t-1,\,T_m))\, \D t,
\label{eqn:Eqn_Rate_um}
\end{equation}
\begin{proof}
	See Appendix~\ref{app:App_MBS_User_Rate_IBFD}
\end{proof}
The coverage expression for M-BS-user case, given by \eqref{eqn:Eqn7}, can be used in \eqref{eqn:Eqn_Rate_um} to obtain the average conditional rates.

\subsection{M-BS to P-BS to User Rate ($\min\left(R_{us},\, R_{sm}\right)$)}
The expectation in the second term in \eqref{eqn:Eqn_Rate} is computed now. For notational simplicity, let  
$\left\{\SIR_{us,sm} > T_{s,b}\right\}\triangleq \{\SIR_{us}>T_s,\,\SIR_{sm}>T_b\}$. On an average, each macro cell is 
assumed to backhaul $n$ small cells, where $n =\lambda_s/\lambda_m$. 
This means that on the backhaul link, the rate to each P-BS will get reduced by a factor of $n$, besides being multiplied 
by $\eta$, which is the amount of bandwidth from $1\,Hz$, that is allocated by the M-BSs for backhauling P-BSs.
\begin{multline}
\E\left[\min\left(R_{us},\, R_{sm}\right)\mid \SIR_{us,sm} > T_{s,b}\right] = \\ \frac{1}{\Pr(\SIR_{us,sm} > T_{s,b})}\int\limits_{t > 0} \Pr\left( \SIR_{us} > \max(2^{t}-1,\,T_s),\,\, \SIR_{sm} > \max(2^{\frac{nt}{\eta}}-1,\,T_b)\right)\, \D t.
\end{multline}
\begin{proof}
	See Appendix~\ref{app:App_MBS_PBS_User_Rate}
\end{proof}
The average conditional rate could be similarly calculated for the FDD system keeping note of the fact that the bandwidth gets split into 
$0.5\,Hz$ each for the backhaul and the access links, which essentially, at least theoretically, must halve the rates for a FDD system in 
comparison to a IBFD system.

\section{Numerical Results}
\label{sec:Sec_NumResults}
This section numerically computes the coverage expressions provided in the previous sections and compares them with 
Monte Carlo simulations. The  parameters used for Monte-Carlo simulation are the same as mentioned in section~\ref{subsec:Subsec1_BWAlloc}.
Simulation is done with PPPs $\Phi_s$ and $\Phi_m$ on an area of $60 \times 60$ square units with $14400$ and $3600$ nodes, respectively. 
All simulations are shown with the self-interference factor $\beta= 0$ dB, path loss exponent for the M-BS tier, $\alpha_m = 2.8$ and for the P-BS tier, $\al_s = 4$, unless mentioned otherwise. Transmit powers of M-BS and P-BS are proportionally considered as $P_m=150$ and $P_s=1$ in accordance with powers of 
$46$ dBm and $24$ dBm respectively for wide-area and local-area BS \cite{3gpp.36.104}. Bias towards M-BS $B_m = 0$ dB, unless 
mentioned otherwise.
 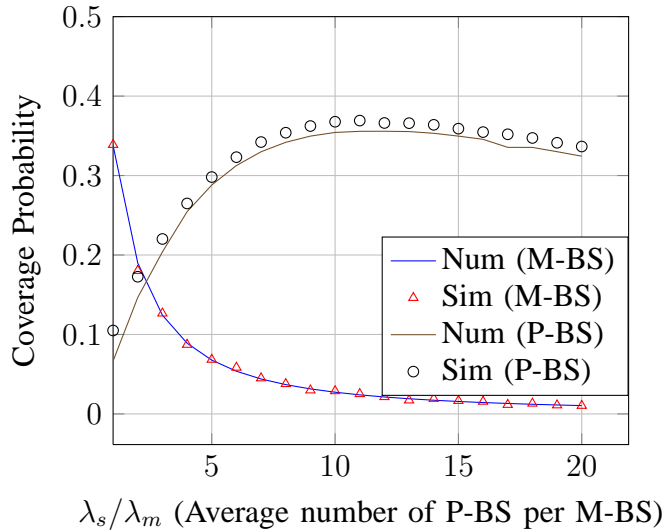
\begin{figure}[!ht]
 \centering 
 \begin{tikzpicture}
 \begin{axis}[scale=1.0,
 grid = both,
 legend cell align=left,
 legend style={align=left,at={(1.0,0.30)},anchor=east},
 xmin = 1,ymax = 0.5,
 xlabel ={$\lambda_s/\lambda_m$ (Average number of P-BS per M-BS) },
 ylabel = {Coverage Probability}
 ]
 
\addplot+[no marks, x] 
coordinates{
(1.000000, 0.337949) (2.000000, 0.189936) (3.000000, 0.124253) (4.000000, 0.089031) (5.000000, 0.067732) (6.000000, 0.053746) (7.000000, 0.043997) (8.000000, 0.036887) (9.000000, 0.031515) (10.000000, 0.027340) (11.000000, 0.024019) (12.000000, 0.021326) (13.000000, 0.019106) (14.000000, 0.017251) (15.000000, 0.015680) (16.000000, 0.014337) (17.000000, 0.013178) (18.000000, 0.012169) (19.000000, 0.011284) (20.000000, 0.010503) 
};\addlegendentry{Num (M-BS)}

\addplot+[only marks, mark=triangle, x] 
coordinates{
(1.000000, 0.338800) (2.000000, 0.180200) (3.000000, 0.126600) (4.000000, 0.087200) (5.000000, 0.068400) (6.000000, 0.058400) (7.000000, 0.045000) (8.000000, 0.037800) (9.000000, 0.029800) (10.000000, 0.028800) (11.000000, 0.025200) (12.000000, 0.021400) (13.000000, 0.017200) (14.000000, 0.019200) (15.000000, 0.016600) (16.000000, 0.015600) (17.000000, 0.011600) (18.000000, 0.013200) (19.000000, 0.011000) (20.000000, 0.010400) 
};\addlegendentry{Sim (M-BS)}

\addplot+[no marks, x] 
coordinates{
(1.000000, 0.066645) (2.000000, 0.146568) (3.000000, 0.203901) (4.000000, 0.254482) (5.000000, 0.288089) (6.000000, 0.312637) (7.000000, 0.329956) (8.000000, 0.341848) (9.000000, 0.349649) (10.000000, 0.354250) (11.000000, 0.355661) (12.000000, 0.355768) (13.000000, 0.355343) (14.000000, 0.353190) (15.000000, 0.349630) (16.000000, 0.345593) (17.000000, 0.335470) (18.000000, 0.335439) (19.000000, 0.330018) (20.000000, 0.324375) 
};\addlegendentry{Num (P-BS)}

\addplot+[only marks, mark=o, x] 
coordinates{
(1.000000, 0.105000) (2.000000, 0.172667) (3.000000, 0.220180) (4.000000, 0.264880) (5.000000, 0.298120) (6.000000, 0.323140) (7.000000, 0.342160) (8.000000, 0.353920) (9.000000, 0.362320) (10.000000, 0.367740) (11.000000, 0.369180) (12.000000, 0.366160) (13.000000, 0.365920) (14.000000, 0.363840) (15.000000, 0.359040) (16.000000, 0.354700) (17.000000, 0.351960) (18.000000, 0.347300) (19.000000, 0.341240) (20.000000, 0.336560) 
};\addlegendentry{Sim (P-BS)}

\end{axis}
 \end{tikzpicture}
 \caption{Coverage Probability vs. Small Cell Density. ($T_m = T_s = T_b = -10\,$dB. $B_s = 22$ dB $\lambda_m = 1, \al_m = 2.8, \al_s = 4$)}
 \label{fig:Fig_PcFD_Ls}
 \end{figure}

The coverage probability is plotted with respect to different parameters in  Fig.~\ref{fig:Fig_PcFD_Ls}
and Fig.~\ref{fig:Fig_PcFD_Ts}. A close match between the simulations and the numerical evaluation of the theoretical 
expressions is seen. This establishes the 
validity of the derived analytical framework, that is tractable and quick in computing the network coverage trends in the proposed IBFD self-backhauling network.
The $\SIR$ for a typical user in a IBFD self-backhauling network is far lesser than that of its FDD counterpart, which results in 
much less coverage for a IBFD network. This is primarily because of the inter-tier interference in addition to the intra-tier 
interferers (intra-tier interference present in FDD network too) in an IBFD network. More biasing towards 
the P-BS tier requires more backhauling on the same spectrum, eventually resulting in increased interference to the access links.
 \begin{figure}
 \centering 
 \begin{tikzpicture}
 \begin{axis}[scale=1.0,
 grid = both,
 legend cell align=left,
 legend style={align=left,at={(1.0,1.0)},anchor=north east},
 xlabel ={$T_s$ dB (P-BS $\SIR$ Threshold) },
 ylabel = {Coverage Probability}
 ]
 
\addplot+[no marks]
coordinates{
(-10.000000, 0.228048) (-4.559320, 0.165226) (-2.218487, 0.137662) (-0.705811, 0.108540) (0.413927, 0.094015) (1.303338, 0.082764) (2.041200, 0.075740) (2.671717, 0.070272) (3.222193, 0.065845) (3.710679, 0.062161) (4.149733, 0.060714) (4.548449, 0.058008) (4.913617, 0.055638) (5.250448, 0.053540) (5.563025, 0.051665) (5.854607, 0.048697) (6.127839, 0.047165) (6.384893, 0.045767) (6.627578, 0.044484) (6.857417, 0.043302) (7.075702, 0.043477) (7.283538, 0.042460) (7.481880, 0.041512) (7.671559, 0.040624) (7.853298, 0.039792) (8.027737, 0.039010) (8.195439, 0.038272) (8.356906, 0.037575) (8.512583, 0.036915) (8.662873, 0.036289) (8.808136, 0.035693) (8.948697, 0.035127) (9.084850, 0.034586) (9.216865, 0.034070) (9.344985, 0.033576) (9.469433, 0.033104) (9.590414, 0.032651) (9.708116, 0.032216) (9.822712, 0.031788) (9.934362, 0.031382)  
};\addlegendentry{Num (P-BS)}

\addplot+[only marks, mark=triangle]
coordinates{
(-10.000000, 0.273600) (-4.559320, 0.165300) (-2.218487, 0.136400) (-0.705811, 0.108900) (0.413927, 0.103100) (1.303338, 0.090700) (2.041200, 0.082200) (2.671717, 0.083600) (3.222193, 0.077300) (3.710679, 0.070000) (4.149733, 0.065100) (4.548449, 0.062900) (4.913617, 0.061200) (5.250448, 0.061300) (5.563025, 0.054500) (5.854607, 0.053700) (6.127839, 0.048700) (6.384893, 0.054300) (6.627578, 0.055200) (6.857417, 0.048900) (7.075702, 0.048200) (7.283538, 0.047100) (7.481880, 0.045200) (7.671559, 0.042900) (7.853298, 0.044700) (8.027737, 0.041400) (8.195439, 0.043000) (8.356906, 0.042200) (8.512583, 0.040700) (8.662873, 0.042600) (8.808136, 0.038300) (8.948697, 0.034900) (9.084850, 0.037200) (9.216865, 0.038300) (9.344985, 0.038000) (9.469433, 0.036100) (9.590414, 0.038300) (9.708116, 0.034000) (9.822712, 0.032300) (9.934362, 0.033900)  
 };\addlegendentry{Sim (P-BS)}
 
\addplot+[no marks]
coordinates{
(-10.000000, 0.317079) (-4.559320, 0.254257) (-2.218487, 0.226693) (-0.705811, 0.197571) (0.413927, 0.183046) (1.303338, 0.171796) (2.041200, 0.164772) (2.671717, 0.159303) (3.222193, 0.154876) (3.710679, 0.151192) (4.149733, 0.149745) (4.548449, 0.147039) (4.913617, 0.144669) (5.250448, 0.142571) (5.563025, 0.140696) (5.854607, 0.137729) (6.127839, 0.136197) (6.384893, 0.134798) (6.627578, 0.133516) (6.857417, 0.132333) (7.075702, 0.132508) (7.283538, 0.131491) (7.481880, 0.130543) (7.671559, 0.129656) (7.853298, 0.128824) (8.027737, 0.128041) (8.195439, 0.127303) (8.356906, 0.126606) (8.512583, 0.125946) (8.662873, 0.125320) (8.808136, 0.124725) (8.948697, 0.124158) (9.084850, 0.123617) (9.216865, 0.123101) (9.344985, 0.122608) (9.469433, 0.122135) (9.590414, 0.121682) (9.708116, 0.121248) (9.822712, 0.120820) (9.934362, 0.120414)  
};\addlegendentry{Num (Overall)}

\addplot+[only marks, mark=oplus]
coordinates{
(-10.000000, 0.367000) (-4.559320, 0.253300) (-2.218487, 0.225500) (-0.705811, 0.206000) (0.413927, 0.191900) (1.303338, 0.178300) (2.041200, 0.172400) (2.671717, 0.172300) (3.222193, 0.162600) (3.710679, 0.152600) (4.149733, 0.153200) (4.548449, 0.152100) (4.913617, 0.149800) (5.250448, 0.147700) (5.563025, 0.138500) (5.854607, 0.141600) (6.127839, 0.139100) (6.384893, 0.141600) (6.627578, 0.144200) (6.857417, 0.139900) (7.075702, 0.138000) (7.283538, 0.135100) (7.481880, 0.132300) (7.671559, 0.130300) (7.853298, 0.130800) (8.027737, 0.130200) (8.195439, 0.133800) (8.356906, 0.128900) (8.512583, 0.128200) (8.662873, 0.130000) (8.808136, 0.127800) (8.948697, 0.117800) (9.084850, 0.126600) (9.216865, 0.124200) (9.344985, 0.126900) (9.469433, 0.121600) (9.590414, 0.125100) (9.708116, 0.119500) (9.822712, 0.114600) (9.934362, 0.122400)  
 };\addlegendentry{Sim (Overall)}
 
 \end{axis}
 \end{tikzpicture}
\caption{Coverage Probability vs. P-BS $\SIR$ Threshold. ($T_m = T_b = -10\,$dB. $B_s = 22$ dB, $\al_m=2.8, \al_s=4$)}
 \label{fig:Fig_PcFD_Ts}
 \end{figure}
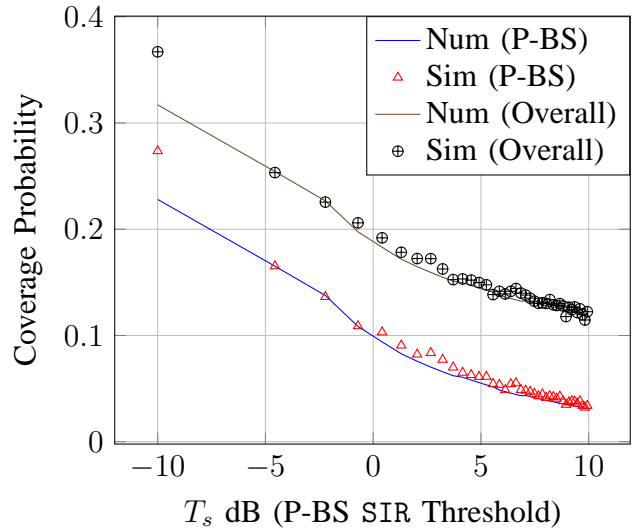
\par
Plots of Fig.~\ref{fig:Fig_PcComp_Ts} and Fig.~\ref{fig:Fig_PcComp_Ls} show the coverage variation versus the P-BS $\SIR$ threshold and ratio of densities of P-BS and M-BS. As expected, the coverage for both IBFD and FDD cases decreases with increasing $T_s$. As $T_s$ is increased, 
users associated with the P-BS do not get sufficient $\SIR$ for coverage. This implies the coverage mostly corresponds to that 
provided by the M-BS and hence at large values of $T_s$ the two curves in Fig.~\ref{fig:Fig_PcComp_Ts} approach each other.
\begin{center}
\centering
\begin{minipage}{.45\linewidth}
 \begin{tikzpicture}
 \begin{axis}[scale=0.8,
 grid = both,
 legend cell align=left,
 legend style={align=left,at={(1.0,1.0)},anchor=north east},
 xlabel ={$T_s$ dB (P-BS $\SIR$ Threshold) },
 ylabel = {Coverage Probability}
 ]
 

\addplot+[]
coordinates{
(-10.000000, 0.317079) (-4.559320, 0.254257) (-2.218487, 0.226693) (-0.705811, 0.197571) (0.413927, 0.183046) (1.303338, 0.171796) (2.041200, 0.164772) (2.671717, 0.159303) (3.222193, 0.154876) (3.710679, 0.151192) (4.149733, 0.149745) (4.548449, 0.147039) (4.913617, 0.144669) (5.250448, 0.142571) (5.563025, 0.140696) (5.854607, 0.137729) (6.127839, 0.136197) (6.384893, 0.134798) (6.627578, 0.133516) (6.857417, 0.132333) (7.075702, 0.132508) (7.283538, 0.131491) (7.481880, 0.130543) (7.671559, 0.129656) (7.853298, 0.128824) (8.027737, 0.128041) (8.195439, 0.127303) (8.356906, 0.126606) (8.512583, 0.125946) (8.662873, 0.125320) (8.808136, 0.124725) (8.948697, 0.124158) (9.084850, 0.123617) (9.216865, 0.123101) (9.344985, 0.122608) (9.469433, 0.122135) (9.590414, 0.121682) (9.708116, 0.121248) (9.822712, 0.120820) (9.934362, 0.120414)
 };\addlegendentry{IBFD}

\addplot+[]
coordinates{
(-10.000000, 0.763400) (-2.218487, 0.618533) (0.413927, 0.547080) (2.041200, 0.470680) (3.222193, 0.422200) (4.149733, 0.393240) (4.913617, 0.367400) (5.563025, 0.348920) (6.127839, 0.334280) (6.627578, 0.325040) (7.075702, 0.313480) (7.481880, 0.303680) (7.853298, 0.294240) (8.195439, 0.289200) (8.512583, 0.283200) (8.808136, 0.275680) (9.084850, 0.270000) (9.344985, 0.268040) (9.590414, 0.261133) (9.822712, 0.266000) 
 };\addlegendentry{FDD}
 \end{axis}
\end{tikzpicture}
\captionsetup{}

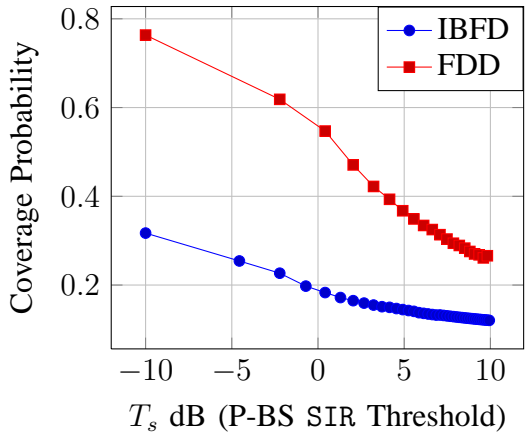
\captionof{figure}{$T_m=T_b=-10\,dB,\,\lambda_s=4\lambda_m,\,B_s=22\,$dB. As expected, coverage decreases with increasing $\SIR$ thresholds. The two curves converge asymptotically as increasing $T_s$ beyond a certain range results in a virtually macro-only network.}
\label{fig:Fig_PcComp_Ts}
\end{minipage}%
\hspace{.05\linewidth}
\begin{minipage}{.45\linewidth}
 \begin{tikzpicture}
 \begin{axis}[scale=0.8,
 grid = both,
 legend cell align=left,
 legend style={align=left,at={(1.0,0.75)},anchor=north east},
 xlabel ={$\lambda_s/\lambda_m$ (P-BSs per M-BS) },
 ylabel = {Coverage Probability}
 ]
 

\addplot+
coordinates{
(1.000000, 0.404594) (2.000000, 0.336504) (3.000000, 0.328154) (4.000000, 0.343513) (5.000000, 0.355821) (6.000000, 0.366383) (7.000000, 0.373953) (8.000000, 0.378735) (9.000000, 0.381164) (10.000000, 0.381590) (11.000000, 0.379680) (12.000000, 0.377094) (13.000000, 0.374449) (14.000000, 0.370441) (15.000000, 0.365310) (16.000000, 0.359930) (17.000000, 0.348648) (18.000000, 0.347608) (19.000000, 0.341302) (20.000000, 0.334878) (21.000000, 0.328390) (22.000000, 0.321878) (23.000000, 0.315381) (24.000000, 0.308927) (25.000000, 0.302536) (26.000000, 0.296228) (27.000000, 0.290017) (28.000000, 0.283914) (29.000000, 0.277935) (30.000000, 0.269681) (31.000000, 0.263923) (32.000000, 0.258303) (33.000000, 0.255258) (34.000000, 0.249918) (35.000000, 0.244712) (36.000000, 0.239631) (37.000000, 0.234689) (38.000000, 0.229879) (39.000000, 0.225195) (40.000000, 0.220638)  
 };\addlegendentry{IBFD $B_s = 22$ dB}
 
 \addplot+
 coordinates{
(1.000000, 0.184200) (2.000000, 0.216467) (3.000000, 0.252160) (4.000000, 0.284120) (5.000000, 0.311360) (6.000000, 0.334320) (7.000000, 0.349320) (8.000000, 0.361160) (9.000000, 0.367840) (10.000000, 0.372600) (11.000000, 0.373520) (12.000000, 0.373680) (13.000000, 0.371320) (14.000000, 0.368840) (15.000000, 0.363400) (16.000000, 0.358200) (17.000000, 0.352760) (18.000000, 0.347960) (19.000000, 0.342920) (20.000000, 0.339920) (21.000000, 0.334640) (22.000000, 0.328960) (23.000000, 0.324480) (24.000000, 0.318440) (25.000000, 0.313160) (26.000000, 0.308640) (27.000000, 0.300520) (28.000000, 0.295400) (29.000000, 0.289120) (30.000000, 0.281760) (31.000000, 0.274360) (32.000000, 0.270920) (33.000000, 0.264000) (34.000000, 0.260200) (35.000000, 0.255880) (36.000000, 0.253720) (37.000000, 0.250120) (38.000000, 0.245680) (39.000000, 0.241000) (40.000000, 0.233600) 
 };\addlegendentry{IBFD $B_s = 34$ dB}
 
\addplot+
coordinates{
(1.000000, 0.808000) (2.000000, 0.788880) (3.000000, 0.776240) (4.000000, 0.764680) (5.000000, 0.757080) (6.000000, 0.752536) (7.000000, 0.750032) (8.000000, 0.749136) (9.000000, 0.748888) (10.000000, 0.748616) (11.000000, 0.747608) (12.000000, 0.746848) (13.000000, 0.745552) (14.000000, 0.743904) (15.000000, 0.742736) (16.000000, 0.741696) (17.000000, 0.740824) (18.000000, 0.740392) (19.000000, 0.740304) (20.000000, 0.740104) (21.000000, 0.740544) (22.000000, 0.740688) (23.000000, 0.740576) (24.000000, 0.740456) (25.000000, 0.740792) (26.000000, 0.740288) (27.000000, 0.740736) (28.000000, 0.741312) (29.000000, 0.741976) (30.000000, 0.742352) (31.000000, 0.743144) (32.000000, 0.742320) (33.000000, 0.741752) (34.000000, 0.741392) (35.000000, 0.741112) (36.000000, 0.740688) (37.000000, 0.741520) (38.000000, 0.740512) (39.000000, 0.739800) (40.000000, 0.735000) 
 };\addlegendentry{FDD}
 
 \end{axis}
 \end{tikzpicture}
\captionsetup{}

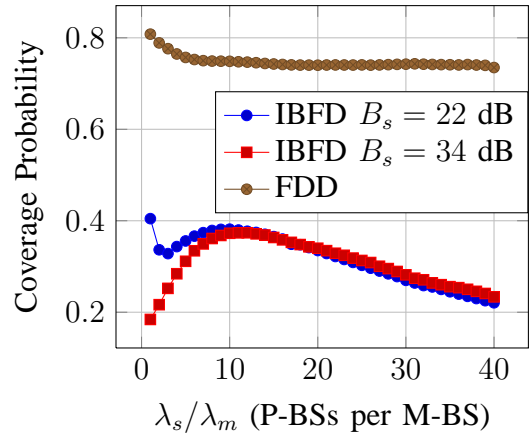
\captionof{figure}{$T_m = T_b = T_s = -10\,$dB. With higher $B_s$, there is an optimal P-BS density achieving maximum coverage. At higher biasing, the P-BS density should be commensurate with the $B_s$ values so as to fully utilize the biasing effect.}
\label{fig:Fig_PcComp_Ls}
\end{minipage}%
\end{center}
\vspace{12px}
\par
In Fig.~\ref{fig:Fig_PcComp_Ls}, the FDD coverage curve is in accordance with the findings in \cite{Jeff}, in that the coverage remains almost constant with increasing 
density of P-BSs. For the IBFD curve, the findings are different. For high biasing towards P-BS, there is an optimal density that maximizes the coverage, whereas for reasonably lower $B_s$ coverage decreases with increasing P-BS density. 
The reason is not very apparent by the total coverage plot of Fig.~\ref{fig:Fig_PcComp_Ls}, but only by inspecting the coverage within backhaul and access layers. 
It is the coverage under P-BSs that gives the shape of the high $B_s$ plot in Fig.~\ref{fig:Fig_PcComp_Ls}. 
Coverage under P-BS is composed of two probabilities--user coverage under P-BS and the P-BS coverage under a 
backhauling M-BS as shown in Fig.~\ref{fig:Fig_PcM_Bs2400}. The plot in Fig.~\ref{fig:Fig_PcM_Bs2400} shows the individual coverage probabilities of M-BS to P-BS (backhaul) and P-BS to user (access) links with varying P-BS density to gain insight into the 
behavior of the coverage plot in Fig.~\ref{fig:Fig_PcComp_Ls}. 
These plots bring out fundamental scaling trends in such a self-backhauling network. 
They show that the net coverage under P-BS increases with P-BS density till an optimum is reached. This is because during this 
increase in density, effective bias towards the P-BS increases and more users associate and subsequently get covered under 
P-BSs, albeit with lower $\SIR$.
\begin{center}
\centering
\begin{minipage}{.45\linewidth}
 \begin{tikzpicture}
 \begin{axis}[scale=0.9,
 grid = both,
 legend cell align=left,
 legend style={align=left,at={(1.0,0.82)},anchor=north east},
 xlabel ={$\lambda_s/\lambda_m$ (P-BSs per M-BS) },
 ylabel = {Coverage Probability}
 ]
 \addplot
coordinates{
(1.000000, 0.718400) (2.000000, 0.732000) (3.000000, 0.720760) (4.000000, 0.710680) (5.000000, 0.693360) (6.000000, 0.674880) (7.000000, 0.656480) (8.000000, 0.635520) (9.000000, 0.617320) (10.000000, 0.597960) (11.000000, 0.583440) (12.000000, 0.568160) (13.000000, 0.557520) (14.000000, 0.546320) (15.000000, 0.532400) (16.000000, 0.518320) (17.000000, 0.505320) (18.000000, 0.489320) (19.000000, 0.475440) (20.000000, 0.464160) (21.000000, 0.451360) (22.000000, 0.438160) (23.000000, 0.428200) (24.000000, 0.420360) (25.000000, 0.410400) (26.000000, 0.402160) (27.000000, 0.394880) (28.000000, 0.384000) (29.000000, 0.373200) (30.000000, 0.365840) (31.000000, 0.359560) (32.000000, 0.350640) (33.000000, 0.346560) (34.000000, 0.337960) (35.000000, 0.331200) (36.000000, 0.320080) (37.000000, 0.312680) (38.000000, 0.304880) (39.000000, 0.295267) (40.000000, 0.295200) 
 };\addlegendentry{Backhaul}
\addplot
coordinates{
(1.000000, 0.141600) (2.000000, 0.243800) (3.000000, 0.319200) (4.000000, 0.394520) (5.000000, 0.454880) (6.000000, 0.504880) (7.000000, 0.547880) (8.000000, 0.584880) (9.000000, 0.614240) (10.000000, 0.640360) (11.000000, 0.665720) (12.000000, 0.685360) (13.000000, 0.703080) (14.000000, 0.718800) (15.000000, 0.732480) (16.000000, 0.742520) (17.000000, 0.752720) (18.000000, 0.762000) (19.000000, 0.770800) (20.000000, 0.778200) (21.000000, 0.786960) (22.000000, 0.793280) (23.000000, 0.799920) (24.000000, 0.804280) (25.000000, 0.809480) (26.000000, 0.810880) (27.000000, 0.816400) (28.000000, 0.820000) (29.000000, 0.824640) (30.000000, 0.828760) (31.000000, 0.833240) (32.000000, 0.835400) (33.000000, 0.837560) (34.000000, 0.838560) (35.000000, 0.840400) (36.000000, 0.845080) (37.000000, 0.848800) (38.000000, 0.851840) (39.000000, 0.858000) (40.000000, 0.855000)  
 };\addlegendentry{Access} 
\addplot[mark = triangle]
coordinates{
(1.000000, 0.105400) (2.000000, 0.175400) (3.000000, 0.221240) (4.000000, 0.265520) (5.000000, 0.299120) (6.000000, 0.321960) (7.000000, 0.339000) (8.000000, 0.349120) (9.000000, 0.354920) (10.000000, 0.356560) (11.000000, 0.361360) (12.000000, 0.362880) (13.000000, 0.364600) (14.000000, 0.366560) (15.000000, 0.365040) (16.000000, 0.360840) (17.000000, 0.356520) (18.000000, 0.350520) (19.000000, 0.344160) (20.000000, 0.339280) (21.000000, 0.333520) (22.000000, 0.326360) (23.000000, 0.321240) (24.000000, 0.316840) (25.000000, 0.310360) (26.000000, 0.303920) (27.000000, 0.300760) (28.000000, 0.293800) (29.000000, 0.287360) (30.000000, 0.283600) (31.000000, 0.281520) (32.000000, 0.275240) (33.000000, 0.273120) (34.000000, 0.266280) (35.000000, 0.260920) (36.000000, 0.253600) (37.000000, 0.248640) (38.000000, 0.242840) (39.000000, 0.238400) (40.000000, 0.235200)  
 };\addlegendentry{Net P-BS} 
 \end{axis}
\end{tikzpicture}
\captionsetup{}
\captionof{figure}{$T_m=T_b=T_s=-10\,\text{ dB},\,B_s=34\,$dB. Net P-BS coverage curve is shaped by two probabilities, P-BS to user and M-BS to P-BS coverage.}
\label{fig:Fig_PcM_Bs2400}
\end{minipage}
\hspace{.05\linewidth}
\begin{minipage}{.45\linewidth}
	\begin{tikzpicture}
	\begin{axis}[scale=0.9,
	grid = both,
	legend cell align=left,
	legend style={align=left,at={(1.0,1.0)},anchor=north east},
	xlabel ={$\beta$ dB (Self-interference factor) },
	ylabel = {Coverage Probability}
	]
	
	\addplot
	coordinates{
	(-20.000000, 0.369133) (-16.000000, 0.366062) (-12.000000, 0.363683) (-8.000000, 0.361531) (-4.000000, 0.359925) (0.000000, 0.358713) (4.000000, 0.357813) (8.000000, 0.357265) (12.000000, 0.356749) (16.000000, 0.355073) (20.000000, 0.350413) (24.000000, 0.342928) (28.000000, 0.327456) (32.000000, 0.310698) (36.000000, 0.278933) 
	};
	
	\end{axis}
	\end{tikzpicture}
	\captionsetup{}
	\captionof{figure}{$T_m=T_b=T_s=-10\,$dB, $B_s = 22$ dB, $\lambda_s = 4 \lambda_m$. Coverage with varying $\beta$. As expected, coverage reduces with reducing self-interference cancellation capability.} 
	\label{fig:Fig_Pc_Beta}
\end{minipage}
 \end{center}
Moreover, there is scope for the M-BSs to cater to more P-BSs for backhauling. The result is an increase in coverage.
Beyond the optimum coverage point, user coverage under P-BSs starts to stagnate but the backhauling coverage drops steeply. 
Stagnation in P-BS to user coverage is due to the fact that at high P-BS density, users mostly associate with P-BS. Then, the 
network behaves as if a single tier network with increasing BS density which is know to be constant \cite{Jeff}. 
On the other hand the backhaul coverage drops due to the increasing interference that the access links of the P-BS pose to the backhaul links of the P-BS. This effectively results in an overall decrease in the coverage under P-BSs. The same is not true of the FDD counterpart of such a network. In FDD case, backhauling links do not interfere with the access links. With increasing P-BS density, an increasing P-BS coverage balances a declining M-BS coverage to give an almost constant net coverage. 
\par
As expected, Fig.~\ref{fig:Fig_Pc_Beta} shows the degradation of coverage with increasing self-interference factor at the P-BS.
\begin{center}
	\centering
	\begin{minipage}{.45\linewidth}
  \begin{tikzpicture}
 \begin{axis}[scale=0.9,
 grid = both,
 legend cell align=left,
 xmin = -20,
 legend style={align=left,at={(1.0,0.75)},anchor=north east},
 xlabel ={$B_s$ dB (P-BS Bias) },
 ylabel = {Coverage Probability}
 ]
 

\addplot+[]
coordinates{
(-10.000000, 0.732900) (-8.000000, 0.737696) (-6.000000, 0.738653) (-4.000000, 0.738185) (-2.000000, 0.733472) (0.000000, 0.724332) (2.000000, 0.709600) (4.000000, 0.688648) (6.000000, 0.661428) (8.000000, 0.628936) (10.000000, 0.591740) (12.000000, 0.551612) (14.000000, 0.510848) (16.000000, 0.471364) (18.000000, 0.434520) (20.000000, 0.402640) (22.000000, 0.375668) (24.000000, 0.353792) (26.000000, 0.336340) (28.000000, 0.322964) (30.000000, 0.312216) (32.000000, 0.303912) (34.000000, 0.297392) (36.000000, 0.293072) (38.000000, 0.289940) (40.000000, 0.286988) (42.000000, 0.285072) (44.000000, 0.283788) (46.000000, 0.281824) (48.000000, 0.280336) (50.000000, 0.280252) (52.000000, 0.280476) (54.000000, 0.280997) (56.000000, 0.283305) (58.000000, 0.285522) (60.000000, 0.289900) 
 };\addlegendentry{IBFD}
 
\addplot+[mark=triangle]
coordinates{
(-10.000000, 0.812700) (-8.000000, 0.811418) (-6.000000, 0.809867) (-4.000000, 0.808487) (-2.000000, 0.806940) (0.000000, 0.805812) (2.000000, 0.804896) (4.000000, 0.803592) (6.000000, 0.801840) (8.000000, 0.799652) (10.000000, 0.795948) (12.000000, 0.790644) (14.000000, 0.784892) (16.000000, 0.778476) (18.000000, 0.771652) (20.000000, 0.765772) (22.000000, 0.761044) (24.000000, 0.756104) (26.000000, 0.752064) (28.000000, 0.748676) (30.000000, 0.745192) (32.000000, 0.741836) (34.000000, 0.739984) (36.000000, 0.738740) (38.000000, 0.738204) (40.000000, 0.738824) (42.000000, 0.740476) (44.000000, 0.740892) (46.000000, 0.740952) (48.000000, 0.740304) (50.000000, 0.738808) (52.000000, 0.736640) (54.000000, 0.735833) (56.000000, 0.736985) (58.000000, 0.738129) (60.000000, 0.744600)  
};\addlegendentry{FDD}

 \end{axis}
 \end{tikzpicture}
  \captionsetup{}
\captionof{figure}{$T_m=T_b=T_s=-10\,$dB,$\,\lambda_s=4\lambda_m$. The coverage falls when users are biased to associate to P-BS, even with lesser $\SIR$ than they see with the M-BS. This stagnates at a point where almost all users are associated to P-BS.} 
 \label{fig:Fig_PcComp_Bs}
\end{minipage}
\hspace{.05\linewidth}
\begin{minipage}{.45\linewidth}
	\begin{tikzpicture}
	\begin{axis}[scale=0.9,
	grid = both,
	legend cell align=left,
	legend style={align=left,at={(1.0,0.75)},anchor=north east},
	xlabel ={$\al_s$ (P-BS Tier pathloss exponent) },
	ylabel = {Coverage Probability}
	]
	
	\addplot+[mark=triangle]
	coordinates{
(2.100000, 0.350500) (2.300000, 0.337878) (2.500000, 0.332227) (2.700000, 0.327587) (2.900000, 0.327000) (3.100000, 0.330680) (3.300000, 0.337080) (3.500000, 0.344180) (3.700000, 0.352720) (3.900000, 0.361820) (4.100000, 0.370820) (4.300000, 0.378800) (4.500000, 0.387280) (4.700000, 0.393780) (4.900000, 0.399933) (5.100000, 0.402500)  
	};\addlegendentry{IBFD}
	
	\addplot+
	coordinates{
	(2.100000, 0.661000) (2.300000, 0.674433) (2.500000, 0.690820) (2.700000, 0.705880) (2.900000, 0.720620) (3.100000, 0.732820) (3.300000, 0.743180) (3.500000, 0.751400) (3.700000, 0.756900) (3.900000, 0.761540) (4.100000, 0.765120) (4.300000, 0.766820) (4.500000, 0.769020) (4.700000, 0.771960) (4.900000, 0.773900) (5.100000, 0.778500) 
		};\addlegendentry{FDD}
	
	 \end{axis}
	 \end{tikzpicture}
	 \captionsetup{}
	 \captionof{figure}{$T_m=T_b=T_s=-10\,$dB,$\,\lambda_s=4\lambda_m$. The coverage increases with increasing pathloss exponent $\al_s$. A larger pathloss exponent is helpful in a heterogeneous network with reasonably high density, as it mitigates interference in a dense network.} 
	 \label{fig:Fig_PcComp_alphaS}
	\end{minipage}
 \end{center}
The plot in Fig.~\ref{fig:Fig_PcComp_Bs} shows that biasing more towards the 
P-BS forces the users to associate with them even when the $\SIR$ received from them is lesser than that from the 
M-BS. This results in decrease in coverage until a point where mostly all users are associated with the P-BS tier only and 
therefore the coverage stagnates. The plot also suggests that the decrease in coverage in the IBFD case is much steeper than in 
the FDD case. This is because in the IBFD case, the P-BS tier receives maximum interference--from other M-BSs as well 
as all the P-BSs. For a user to be biased in associating with a P-BS in a IBFD case is essentially 
forcing it to accept a much weaker $\SIR$ link than in the case of FDD operation. Hence the coverage for a user in IBFD operation degrades much more rapidly than in the FDD case. The plot of Fig.~\ref{fig:Fig_PcComp_alphaS} shows that higher pathloss exponent helps a dense P-BS deployment as it creates virtual cell splitting. The plot shows an initial dip in coverage, but only till $\al_s = \al_m = 2.8$.
\par
The following plots show the variation of average conditional rate of a typical user in a IBFD and FDD self-backhauling network. All rates are calculated 
keeping the bandwidth partitioning parameter $\eta = 0.8$. Since the available bandwidth 
is entirely used by the P-BSs and M-BSs in IBFD network, the rate in IBFD networks, typically tends to twice that of FDD networks. As the interference in IBFD network is more than the conventional FDD network,  the rate is not twice that of the FDD networks.
\begin{center}
\centering
\begin{minipage}{.45\linewidth}
 \begin{tikzpicture}
 \begin{axis}[scale=0.8,
 grid = both,
 legend cell align=left,
 legend style={align=left,at={(1.0,1.0)},anchor=north east},
 xlabel ={$B_s$ dB (P-BS Bias) },
 ylabel = {Covered Rate (b/s/Hz)}
 ]

\addplot+[]
coordinates{
(-10.000000, 0.241766) (-8.000000, 0.241188) (-6.000000, 0.242047) (-4.000000, 0.243296) (-2.000000, 0.245644) (0.000000, 0.248340) (2.000000, 0.252152) (4.000000, 0.257025) (6.000000, 0.262123) (8.000000, 0.267611) (10.000000, 0.273295) (12.000000, 0.277855) (14.000000, 0.281121) (16.000000, 0.283315) (18.000000, 0.283324) (20.000000, 0.280910) (22.000000, 0.276293) (24.000000, 0.269197) (26.000000, 0.259726) (28.000000, 0.248965) (30.000000, 0.237659) (32.000000, 0.226877) (34.000000, 0.217008) (36.000000, 0.207719) (38.000000, 0.199544) (40.000000, 0.192769) (42.000000, 0.186617) (44.000000, 0.181315) (46.000000, 0.177765) (48.000000, 0.175403) (50.000000, 0.173266) (52.000000, 0.171848) (54.000000, 0.170815) (56.000000, 0.169304) (58.000000, 0.167751) (60.000000, 0.165370) 
};\addlegendentry{IBFD}

\addplot+[]
coordinates{
(-10.000000, 0.132073) (-8.000000, 0.133125) (-6.000000, 0.133239) (-4.000000, 0.133587) (-2.000000, 0.133634) (0.000000, 0.134115) (2.000000, 0.134721) (4.000000, 0.136024) (6.000000, 0.137559) (8.000000, 0.139253) (10.000000, 0.140965) (12.000000, 0.142639) (14.000000, 0.143671) (16.000000, 0.144035) (18.000000, 0.144009) (20.000000, 0.143250) (22.000000, 0.141803) (24.000000, 0.139948) (26.000000, 0.137805) (28.000000, 0.135384) (30.000000, 0.133003) (32.000000, 0.130872) (34.000000, 0.128723) (36.000000, 0.126949) (38.000000, 0.125435) (40.000000, 0.124098) (42.000000, 0.122845) (44.000000, 0.122139) (46.000000, 0.121662) (48.000000, 0.121219) (50.000000, 0.120891) (52.000000, 0.120649) (54.000000, 0.120213) (56.000000, 0.119574) (58.000000, 0.118983) (60.000000, 0.118144) 
 };\addlegendentry{FDD}
 
 \end{axis}
\end{tikzpicture}
\captionsetup{}

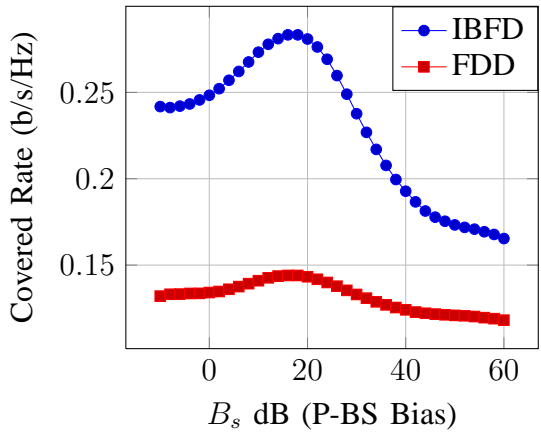
\captionof{figure}{$T_m=T_b=T_s=-10\,$dB,$\,\lambda_s=4\lambda_m$. Covered rate increases higher $B_s$ until an optimal point, after which it reduces as the P-BS density remains constant.}
\label{fig:Fig_RateComp_Bias}
\end{minipage}
\hspace{.05\linewidth}
\begin{minipage}{.45\linewidth}
  \begin{tikzpicture}
 \begin{axis}[scale=0.8,
 grid = both,
 legend cell align=left,
 legend style={align=left,at={(1.0,1.0)},anchor=north east},
 xlabel ={$\lambda_s/\lambda_m$ (P-BSs Per M-BS) },
 ylabel = {Covered Rate (b/s/Hz)}
 ]
 
\addplot+
coordinates{
(1.000000, 0.398447) (2.000000, 0.368032) (3.000000, 0.330348) (4.000000, 0.294226) (5.000000, 0.257096) (6.000000, 0.224775) (7.000000, 0.196024) (8.000000, 0.172231) (9.000000, 0.153208) (10.000000, 0.137994) (11.000000, 0.125198) (12.000000, 0.114950) (13.000000, 0.105942) (14.000000, 0.098336) (15.000000, 0.091561) (16.000000, 0.085658) (17.000000, 0.079889) (18.000000, 0.075344) (19.000000, 0.071127) (20.000000, 0.067539) (21.000000, 0.064502) (22.000000, 0.062176) (23.000000, 0.060034) (24.000000, 0.058326) (25.000000, 0.056713) (26.000000, 0.055097) (27.000000, 0.053977) (28.000000, 0.052816) (29.000000, 0.051934) (30.000000, 0.050364) (31.000000, 0.048825) (32.000000, 0.046985) (33.000000, 0.045204) (34.000000, 0.043011) (35.000000, 0.041899) (36.000000, 0.041023) (37.000000, 0.039815) (38.000000, 0.038959) (39.000000, 0.038250) (40.000000, 0.037731) 
 };\addlegendentry{IBFD}

\addplot+
coordinates{
(1.000000, 0.300911) (2.000000, 0.241885) (3.000000, 0.201217) (4.000000, 0.162945) (5.000000, 0.135635) (6.000000, 0.114160) (7.000000, 0.098727) (8.000000, 0.087261) (9.000000, 0.078004) (10.000000, 0.070503) (11.000000, 0.064366) (12.000000, 0.059389) (13.000000, 0.055111) (14.000000, 0.051571) (15.000000, 0.048537) (16.000000, 0.045810) (17.000000, 0.043253) (18.000000, 0.040864) (19.000000, 0.038798) (20.000000, 0.036957) (21.000000, 0.035206) (22.000000, 0.033404) (23.000000, 0.031884) (24.000000, 0.030252) (25.000000, 0.028647) (26.000000, 0.027300) (27.000000, 0.026444) (28.000000, 0.025494) (29.000000, 0.024685) (30.000000, 0.024022) (31.000000, 0.023422) (32.000000, 0.022651) (33.000000, 0.021982) (34.000000, 0.021343) (35.000000, 0.020674) (36.000000, 0.019931) (37.000000, 0.019267) (38.000000, 0.018796) (39.000000, 0.018312) (40.000000, 0.018300) 
 };\addlegendentry{FDD}

 \end{axis}
 \end{tikzpicture}
 \captionsetup{}

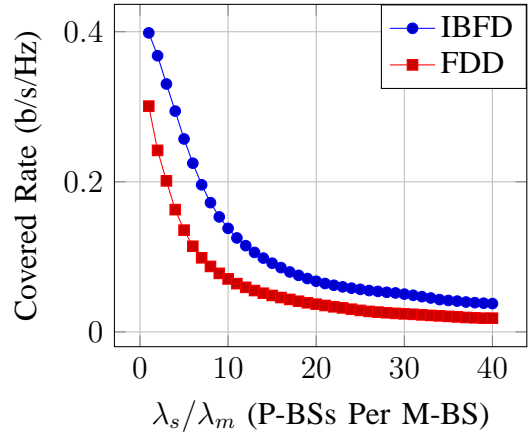
\captionof{figure}{$T_m=T_b=T_s=-10\,$dB,$\,B_s=22\,$dB. Denser P-BSs limit the backhaul bandwidth available per P-BS and the covered rate falls.}
\label{fig:Fig_RateComp_Ls}
\end{minipage}
\end{center}
 \begin{figure}[!ht]
 \centering 
 \begin{tikzpicture}
 \begin{axis}[scale=1.1,
 grid = both,
 legend cell align=left,
 legend style={align=left,at={(1.0,0.85)},anchor=east},
 xlabel ={$\eta$ (Bandwidth sharing factor at M-BS) },
 ylabel = {Covered Rate (b/s/Hz)}
 ]
 
\addplot+[mark = triangle]
coordinates{
(0.001000, 0.755520) (0.151000, 0.680388) (0.301000, 0.596656) (0.451000, 0.508552) (0.601000, 0.412183) (0.751000, 0.324211) (0.901000, 0.224764) 
 };\addlegendentry{Net Rate}

\addplot+[]
coordinates{
(0.001000, 0.755307) (0.151000, 0.652246) (0.301000, 0.545408) (0.451000, 0.433437) (0.601000, 0.314832) (0.751000, 0.205061) (0.901000, 0.084671)  
 };\addlegendentry{M-BS Rate}
 
\addplot+[mark = o]
coordinates{
(0.001000, 0.000213) (0.151000, 0.028916) (0.301000, 0.055295) (0.451000, 0.074007) (0.601000, 0.097808) (0.751000, 0.119550) (0.901000, 0.140093) 
 };\addlegendentry{P-BS Rate}
\end{axis}
 \end{tikzpicture}
 \caption{Covered Rate vs. Bandwidth Shairng at M-BS. ($T_m = T_s = T_b = -10\,$dB. $B_s = 22$ dB), Available bandwidth at the M-BS needs to be segregated into resources used for backhauling P-BS and for direct access to users.}
 \label{fig:Fig_Rate_eta}
 \end{figure}
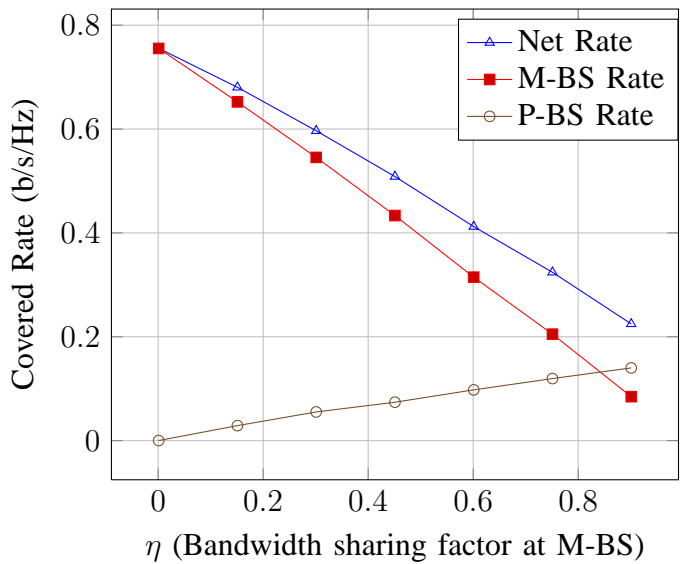

The plots in Fig.~\ref{fig:Fig_RateComp_Bias} and Fig.~\ref{fig:Fig_RateComp_Ls} show the variation of rate with varying $T_s$ and $\lambda_s$. 
As expected, the average normalized rate increases with increasing $T_s$ and decreases with increasing P-BS density. 
In Fig.~\ref{fig:Fig_RateComp_Ls}, increasing P-BS density reduces the backhaul 
bandwidth per P-BS and the rate (which is the minimum over backhaul and access link) also reduces. Thus, the 
interference from the backhaul to the access links as well as the division of bandwidth at the backhauling M-BS, are two major 
limitations in the considered IBFD self-backhauling network.
The plot in Fig.~\ref{fig:Fig_RateComp_Bias} shows that with $\eta=0.8$, there exists a bias point that achieves the maximum 
average rate. Since the density of P-BSs is four times that of M-BSs, there exists a point where all the P-BSs are fully 
utilized to deliver rate to the typical user and hence the shape of the curve. Beyond this point, as the users are forced to 
associate to a weaker $\SIR$ link from the P-BS, the average rate begins to fall. The results obtained in this section indicate two major impediments to achieving the full potential of IBFD self-backhauling 
networks that are \emph{inter-tier interference} from the backhaul to access links and \emph{bandwidth division} at the M-BS to 
accommodate backhauling resources for multiple P-BSs.\\

\section{Conclusion}
\label{sec:Sec_WF}
This work proposed and analyzed a self-backhauling HetNet architecture for IBFD as well as traditional FDD enabled base-stations. 
A tractable and quick-to-compute analytical model for network wide coverage is derived and shown to match 
simulation results. The paper shows that the proposed IBFD self-backhauling network suffers from limitations posed by the inter-tier 
interference and the bandwidth division at the backhauling M-BS. Though IBFD capability helps improve the average 
rates (conditioned on user being covered) by a factor less than double, the coverage in such a network is close to half of its FDD 
counterpart. Analytical framework for exact quantification of coverage under varying parameters such as P-BS density, bias, pathloss exponent, etc. has been derived.
The proposed architecture requires only small cells to be 
IBFD-enabled, which is practically more suitable than IBFD operation on M-BS and user devices owing to their high transmit powers and small form factors, respectively. The paper uses an example IBFD network for clear exposition though similar analysis holds for time-division duplexed (TDD) networks, for instance, by replacing frequencies $f1$ and $f2$ by time-slots $t1$ and $t2$.
\bibliography{Master}

\begin{thebibliography}{10}

\bibitem{Metis}
METIS, ``{Mobile and wireless communications Enablers for the Twenty-twenty
  Information Society},'' tech. rep., {EU 7th Framework Programme Project},
  Mar. 2013.

\bibitem{3gppHetNets}
A.~Damnjanovic, J.~Montojo, Y.~Wei, T.~Ji, T.~Luo, M.~Vajapeyam, T.~Yoo,
  O.~Song, and D.~Malladi, ``A survey on 3gpp heterogeneous networks,'' {\em
  Wireless Communications, IEEE}, vol.~18, pp.~10--21, June 2011.

\bibitem{3gpp.36.932}
ETSI, ``{LTE;Scenarios and requirements for small cell enhancements for E-UTRA
  and E-UTRAN, 3GPP TR 36.932 version 12.1.0 Release 12},'' TR {36.932},
  {European Telecommunications Standards Institute (ETSI)}, Oct. 2014.

\bibitem{Wireless2020}
R.~Schwartz and M.~Rice, ``{Rethinking Small Cell Backhaul},'' tech. rep.,
  {Wireless2020}, July 2012.

\bibitem{nadh2016linearization}
A.~Nadh, J.~Samuel, A.~Sharma, S.~Aniruddhan, and R.~K. Ganti, ``A
  linearization technique for self-interference cancellation in full-duplex
  radios,'' {\em arXiv preprint arXiv:1605.01345}, 2016.

\bibitem{Arjun}
A.~Nadh, A.~Sharma, S.~Aniruddhan, and R.~K. Ganti, ``A taylor series
  approximation technique for self-interference cancellation in full-duplex
  radios,'' in {\em 2016 Twenty Second National Conference on Communication
  (NCC)}, pp.~1--6, March 2016.

\bibitem{Jain:2011:PRF:2030613.2030647}
M.~Jain, J.~I. Choi, T.~Kim, D.~Bharadia, S.~Seth, K.~Srinivasan, P.~Levis,
  S.~Katti, and P.~Sinha, ``Practical, real-time, full duplex wireless,'' in
  {\em Proceedings of the 17th Annual International Conference on Mobile
  Computing and Networking}, MobiCom '11, (New York, NY, USA), pp.~301--312,
  ACM, 2011.

\bibitem{ASabharwal}
M.~Duarte and A.~Sabharwal, ``Full-duplex wireless communications using
  off-the-shelf radios: Feasibility and first results,'' in {\em Signals,
  Systems and Computers (ASILOMAR), 2010 Conference Record of the Forty Fourth
  Asilomar Conference on}, pp.~1558--1562, Nov 2010.

\bibitem{Katti}
S.~Hong, J.~Brand, J.~Choi, M.~Jain, J.~Mehlman, S.~Katti, and P.~Levis,
  ``Applications of self-interference cancellation in 5g and beyond,'' {\em
  Communications Magazine, IEEE}, vol.~52, pp.~114--121, February 2014.

\bibitem{BoyuLi}
B.~Li and P.~Liang, ``Small cell in-band wireless backhaul in massive {MIMO}
  systems: {A} cooperation of next-generation techniques,'' {\em CoRR},
  vol.~abs/1402.2603, 2014.

\bibitem{Marios}
I.~Atzeni and M.~Kountouris, ``Full-duplex mimo small-cell networks:
  Performance analysis,'' {\em arXiv preprint arXiv:1509.05506}, 2015.

\bibitem{GoyalS}
S.~Goyal, P.~Liu, S.~Hua, and S.~Panwar, ``Analyzing a full-duplex cellular
  system,'' in {\em Information Sciences and Systems (CISS), 2013 47th Annual
  Conference on}, pp.~1--6, March 2013.

\bibitem{Hyungjong}
H.~Kim, S.~Lim, H.~Wang, and D.~Hong, ``Optimal power allocation and outage
  analysis for cognitive full duplex relay systems,'' {\em Wireless
  Communications, IEEE Transactions on}, vol.~11, pp.~3754--3765, October 2012.

\bibitem{Barghi}
S.~Barghi, A.~Khojastepour, K.~Sundaresan, and S.~Rangarajan, ``Characterizing
  the throughput gain of single cell mimo wireless systems with full duplex
  radios,'' in {\em Modeling and Optimization in Mobile, Ad Hoc and Wireless
  Networks (WiOpt), 2012 10th International Symposium on}, pp.~68--74, May
  2012.

\bibitem{Hyungsik}
H.~Ju, S.~Lim, D.~Kim, H.~Poor, and D.~Hong, ``Full duplexity in
  beamforming-based multi-hop relay networks,'' {\em Selected Areas in
  Communications, IEEE Journal on}, vol.~30, pp.~1554--1565, September 2012.

\bibitem{AggarwalV}
V.~Aggarwal, M.~Duarte, A.~Sabharwal, and N.~Shankaranarayanan, ``Full- or
  half-duplex? a capacity analysis with bounded radio resources,'' in {\em
  Information Theory Workshop (ITW), 2012 IEEE}, pp.~207--211, Sept 2012.

\bibitem{Riihonen}
T.~Riihonen, S.~Werner, and R.~Wichman, ``Hybrid full-duplex/half-duplex
  relaying with transmit power adaptation,'' {\em Wireless Communications, IEEE
  Transactions on}, vol.~10, pp.~3074--3085, September 2011.

\bibitem{singh2014tractable}
S.~Singh, M.~Kulkarni, A.~Ghosh, and J.~Andrews, ``Tractable model for rate in
  self-backhauled millimeter wave cellular networks,'' 2014.

\bibitem{Ankit}
A.~Sharma, R.~K. Ganti, and J.~K. Milleth, ``Performance analysis of full
  duplex self-backhauling cellular network,'' in {\em 2016 IEEE International
  Conference on Communications (ICC)}, pp.~1--6, May 2016.

\bibitem{Stoyan}
D.~Stoyan, W.~Kendall, J.~Mecke, and D.~Kendall, {\em Stochastic Geometry and
  Its Applications}.
\newblock John Wiley and Sons, 1996.

\bibitem{ganti2012stochastic}
R.~K. Ganti, ``Stochastic geometry and wireless networks,'' 2012.

\bibitem{Harpreet}
H.~Dhillon, R.~, F.~Baccelli, and J.~Andrews, ``Modeling and analysis of k-tier
  downlink heterogeneous cellular networks,'' {\em Selected Areas in
  Communications, IEEE Journal on}, vol.~30, pp.~550--560, April 2012.

\bibitem{HanShin}
H.-S. Jo, Y.~J. Sang, P.~Xia, and J.~Andrews, ``Heterogeneous cellular networks
  with flexible cell association: A comprehensive downlink sinr analysis,''
  {\em Wireless Communications, IEEE Transactions on}, vol.~11, pp.~3484--3495,
  October 2012.

\bibitem{Duarte}
M.~Duarte, C.~Dick, and A.~Sabharwal, ``Experiment-driven characterization of
  full-duplex wireless systems,'' {\em Wireless Communications, IEEE
  Transactions on}, vol.~11, pp.~4296--4307, December 2012.

\bibitem{simon2005digital}
M.~K. Simon and M.-S. Alouini, {\em Digital communication over fading
  channels}, vol.~95.
\newblock John Wiley \& Sons, 2005.

\bibitem{damnjanovic2011survey}
A.~Damnjanovic, J.~Montojo, Y.~Wei, T.~Ji, T.~Luo, M.~Vajapeyam, T.~Yoo,
  O.~Song, and D.~Malladi, ``A survey on 3gpp heterogeneous networks,'' {\em
  Wireless Communications, IEEE}, vol.~18, no.~3, pp.~10--21, 2011.

\bibitem{Airspan1}
P.~T. Airpsan, ``{Small Cell LTE Deployments Tightly Integrating Access and
  Backhaul},'' tech. rep., {Small Cells Americas}, November 2012.

\bibitem{3gpp.36.104}
3GPP, ``{3rd Generation Partnership Project;Technical Specification Group Radio
  Access Network; Evolved Universal Terrestrial Radio Access (E-UTRA);Base
  Station (BS) radio transmission and reception (Release 13)},'' TS {36.104},
  {3rd Generation Partnership Project (3GPP)}, July 2015.

\bibitem{Jeff}
J.~Andrews, F.~Baccelli, and R.~Ganti, ``A tractable approach to coverage and
  rate in cellular networks,'' {\em Communications, IEEE Transactions on},
  vol.~59, pp.~3122--3134, November 2011.

\bibitem{Mathematica}
E.~W. Weisstein, ``Circle-circle intersection, from mathworld--a wolfram web
  resource,''

\end{thebibliography}
\appendices
\section{Joint PDF Of Distance Pair $(r_s,r)$}
\label{app:App1_AB_PDF}
The joint pdf of the distance pair $(r_s,r)$ that characterizes the joint density of the access-backhaul nodes is derived here.
\subsubsection{$\Delta_m \ge 1$}
Considered first is the arrangement as shown in Fig.~\ref{fig:Fig4_PcSc_SysArrangements}. The representations in Fig.~\ref{fig:Fig4_PcSc_SysArrangements} depict cases 
depending on the location of the backhauling M-BS, provided the user associates with the P-BS at a point $r_s$. Parts $(A),\, (B) \text{ and } (C)$ represent cases 
where the backhaul disc (circle with radius $\|r\|$) and the inner macro disc (circle with radius denoted by $OM' = \Delta_m^{-1}r_s^{\al_s / \al_m}$)
\begin{itemize}
 \item do not intersect
 \item have finite intersection area
 \item represent a single disc (i.e. the backhaul disc engulfs the inner macro disc)
\end{itemize}
Following this, the pdf is composed of three sub-parts depending on where the backhauling M-BS is found.
\begin{itemize}
 \item \emph{Case $(A)$: $0 < \|r\| \le \nu_-(r_s,\Delta_m,\al_s,\al_m)$} In this case the density function is given by the void probabilities \cite{Stoyan} of $\Phi_s$ over $\Phi_s \cap B(o,\|r_s\|)^c$ 
 and of $\Phi_m$ over $\Phi_m \cap (B(o,\|\Delta_m^{-1}r_s^{\al_s / \al_m}\|) \cup B(r_s,\|r\|)^c$ i.e.
 \begin{equation}
 \begin{aligned}
  f(r_s,r) &= \frac{\partial F(r_s,r)}{\partial r_s \partial r}\\
  &= \frac{\partial\left( e^{-\lambda_s \pi r_s^2}\, e^{-\lambda_m \pi \left((\Delta_m^{-1}r_s^{\al_s / \al_m})^{2} + r^2\right)}\right)}{\partial r_s \partial r}\\
  &= \displaystyle\frac{4 e^{-\pi\left(r^2\lambda_m + \frac{r_s^{\frac{2\al_s}{\al_m}}}{\Delta_m^2}\lambda_m + r_s^2 \lambda_s \right)} \pi^2 r \lambda_m \left( \frac{r_s^{\frac{2\al_s}{\al_m}}}{\Delta_m^2} \al_s \lambda_m + r_s^2 \al_m \lambda_s\right)} {r_s \al_m}\,\,.
  \label{eqn:Eqn12_density1} 
  \end{aligned}
 \end{equation}
 
 \item \emph{Case $(B)$: $\|r\| \in \nu_-^+(r_s,\Delta_m,\al_s,\al_m)$} This case has a finite intersection area between the 
 backhaul disc and the inner macro disc. Thus the void probabilities and so the density is calculated as follows.
 \begin{equation}
  \begin{aligned}
   f(r_s,r) = \displaystyle\frac{\partial\left( e^{-\lambda_s \pi r_s^2}\, e^{-\lambda_m \left(\pi(\Delta_m^{-1}r_s^{\al_s / \al_m})^{2} + \pi r^2 - \,\operatorname{lens}(M_1,\, M_2) \right)}\right)}{\partial r_s \partial r}\,\,,
   \label{eqn:Eqn13_density2}
  \end{aligned}
 \end{equation}
where $lens(M_1,\,M_2)$ denotes the area of the lens formed between points $M_1$ and $M_2$ of Fig.~\ref{fig:Fig4_PcSc_SysArrangements} (part $(B)$) and is 
given as in \cite{Mathematica}.

\item \emph{Case $(C)$: $ \|r\| > \nu_+(r_s,\Delta_m,\al_s,\al_m)$} This case has the backhaul disc completely engulf the inner macro disc and the density 
is given as follows.
\begin{equation}
 \begin{aligned}
  f(r_s,r) &= \frac{\partial\left( e^{-\lambda_s \pi r_s^2}\, e^{\lambda_m \pi r^2}\right)}{\partial r_s \partial r}\\
  &= 4 \pi ^2 \lambda_m \lambda_s r\, r_s e^{-\pi  \left(\lambda_m r^2+\lambda_s r_s^2\right)}\,\,.
  \label{eqn:Eqn14_density3} 
  \end{aligned}
 \end{equation}
\end{itemize}

\subsubsection{$0 < \Delta_m < 1$}
For this case the radii of the discs depicted in Fig.~\ref{fig:Fig4_PcSc_SysArrangements} change as $r_s/\Delta_m > r_s$. Similar three cases are depicted in 
Fig.~\ref{fig:Fig5_PcSc_SysArrangements_DeltaLess1}.
\begin{itemize}
\item \emph{Case $(A)$: $0 < \|r\| < \mu_-(r_s,\Delta_m,\al_s,\al_m)$} This case is a zero probability case since 
it is already known that there is no M-BS within radius $\|r_s\|^{\al_s / \al_m} \Delta_m^{-1}$.

\item \emph{Case $(B)$: $\|r\| \in \mu_-^+(r_s,\Delta_m,\al_s,\al_m)$} 
Equation \eqref{eqn:Eqn13_density2} could directly be used to give this density function.

\item \emph{Case $(C)$: $ \|r\| \ge \mu_+(r_s,\Delta_m,\al_s,\al_m)$} 
This case is similar to the engulfment case as \emph{Case $(C)$)} for $\Delta_m \ge 1$. 
Hence, the third part of the density function of Equation \eqref{eqn:Eqn14_density3} could directly be used.

\end{itemize}
 \section{Small Cell Coverage Probability}
\label{app:App2_PcSc}
Coverage probability of a user, given it is associated to a small cell is derived here. The coverage probability is denoted by $\Pr\left(\SIR_{us} > T_s,\, \SIR_{sm} > T_b \mid \varepsilon_s\right)$ .
\begin{equation}
 \Pr\left(\SIR_{us} > T_s,\, \SIR_{sm} > T_b \mid \varepsilon_s\right) = \E_{r_s,r}\left[ \underbrace{\Pr\left(\SIR_{us} > T_s,\, \SIR_{sm} > T_b \mid r_s,\,r\right)}_{P_{u,s}^{'f}(T_s,T_b)} \right]
 \label{eqn:Eqn11_Pc_P-BS_FD}
\end{equation}

where, $r_s$ and $r$ denote the points of Fig.~\ref{fig:Fig3_PcSc} and are varied over a region so that the event $\varepsilon_s$ of 
equation \eqref{eqn:Eqn5_Pcu} is always true.

The inner probability term of equation \eqref{eqn:Eqn11_Pc_P-BS_FD} is derived below. In interest of better clarity, the representation of conditioning on points $r_s$ and $r$ is dropped in the following derivation.
 \begin{align*}
  P_{u,s}^{'f}(T_s,T_b) &= \Pr\left(  \frac{P_s g_{o r_s} \PL{r_s}{-\al_s}}{\sum\limits_{z \in A_1^c}P_s g_{oz} \PL{z}{-\al_s} + \sum\limits_{z \in A_2^c}P_m g_{oz} \PL{z}{-\al_m} + P_m g_{or_m} \PL{r_m}{-\al_m} } > T_s, \right. \\
  & \left. \frac{P_m g_{r_s r_m} \PL{r}{-\al_m}}{\sum\limits_{z \in A_1^c}P_s g_{r_sz} \PL{z-r_s}{-\al_s} + \sum\limits_{z \in A_2^c}P_m g_{r_sz} \PL{z-r_s}{-\al_m} + \beta P_s } > T_b   \right) \\
  &\stackrel{(a)}= \Pr\left( g_{o r_s} > k_s \PL{r_s}{\al_s} I_1 ,\, g_{r_s r_m} > k_m \PL{r}{\al_m} I_2 \right)\,, 
  \end{align*}
where $(a)$ results by taking $k_s = T_s/P_s \text{ and } k_m = T_b/P_m$ and $I_1$ and $I_2$ are short notations for interference terms in $\SIR_{us} \text{ and } \SIR_{sm}$ terms. 
Areas $A_1$ and $A_2$ are as defined in \eqref{eqn:Eqn6_Pcus_FD}. Following from 
the result above,
  \begin{align*}
  P_{u,s}^{'f}(T_s,T_b) &\stackrel{}= \E_{I_1,I_2}\left[ \Pr\left( g_{o r_s} > k_s \PL{r_s}{\al_s} I_1 ,\, g_{r_sr_m} > k_m \PL{r}{\al_m} I_2 \mid I_1, I_2  \right) \right] \\
  &\stackrel{(b)}= \E_{I_1,I_2} \left[   \Pr(g_{o r_s} > k_s \PL{r_s}{\al_s}I_1)\, \, \Pr(g_{r_sr_m} > k_m \PL{r}{\al_m}I_2)  \mid I_1, I_2    \right] \\
  &\stackrel{(c)}= \E_{I_1,I_2} \left[  e^{-k_s \PL{r_s}{\al_s}I_1}\, \, e^{-k_m \PL{r}{\al_m}I_2}  \mid I_1, I_2  \right] \\
  &\stackrel{(d)}= \E_{g_{oz},g_{r_sz},\Phi_s}\left[ \prod\limits_{z \in A_1^c} e^{\left(-k_s\PL{r_s}{\al_s} P_s g_{oz} \PL{z}{-\al_s}\right)}e^{\left(-k_m \PL{r}{\al_m}P_s g_{r_sz} \PL{z-r_s}{-\al_s} \right)} \right]
  \end{align*}
  \begin{align*}
  &\qquad \qquad \qquad \E_{g_{oz},g_{or_m},g_{r_sz},\Phi_m}\left[\prod\limits_{z \in A_2^c} e^{\left( -k_s \PL{r_s}{\al_s}P_m (g_{oz} \PL{z}{-\al_m}+g_{or_m} \PL{r_m}{-\al_m})\right)}e^{\left(-k_m \PL{r}{\al_m}P_m g_{r_sz} \PL{z-r_s}{-\al_m}\right)}  \right]\\
  &\qquad \qquad \qquad e^{\left(-k_m \PL{r}{\al_m} \beta P_s\right)}.
  \end{align*}
Assumption of independently fading links gives result to $(b)$. Result in $(c)$ is based on the assumption of fading power being 
exponentially fading with unit mean. Expanding $I_1$, $I_2$ and separating terms belonging to the independent processes 
$\Phi_m \text{ and } \Phi_s$, the result is as given by $(d)$. Simplifying further,       
    \begin{align*}
  P_{u,s}^{'f}(T_s,T_b) &\stackrel{(e)}=\, \E_{g_{oz},g_{r_sz},\Phi_s}\left[ \prod\limits_{z \in A_1^c} e^{\left(-k_s\PL{r_s}{\al_s} P_s g_{oz} \PL{z}{-\al_s}\right)}e^{\left(-k_m \PL{r}{\al_m}P_s g_{r_sz} \PL{z-r_s}{-\al_s} \right)} \right]\\
  &\E_{g_{oz},g_{r_sz},\Phi_m}\left[\prod\limits_{z \in A_2^c} e^{\left( -k_s \PL{r_s}{\al_s}P_m (g_{oz} \PL{z}{-\al_m})\right)}e^{\left(-k_m \PL{r}{\al_m}P_m g_{r_sz} \PL{z-r_s}{-\al_m}\right)}  \right]\\
  & \E_{g_{or_m}}\left[ e^{-k_s \PL{r_s}{\al_s}P_m g_{or_m} \PL{r_m}{-\al_m}} \right] e^{\left(-k_m \PL{r}{\al_m} \beta P_s\right)}\\
  &\stackrel{(f)}=\, \exp\left(-\lambda_s \int\limits_{z \in A_1^c} 1 - \frac{1}{(1+k_s\PL{r_s}{\al_s}P_s\PL{z}{-\al_s})(1+k_m\PL{r}{\al_m}P_s \PL{z-r_s}{-\al_s})}\, \D z \right)\\
  & \exp\left(-\lambda_m \int\limits_{z \in A_2^c} 1 - \frac{1}{(1+k_s\PL{r_s}{\al_s}P_m\PL{z}{-\al_m})(1+k_m\PL{r}{\al_m}P_m \PL{z-r_s}{-\al_m})}\, \D z \right)\\
  &  e^{\left(-k_m \PL{r}{\al_m} \beta P_s\right)} \left(\frac{1}{1+k_s\PL{r_s}{\al_s} P_m \PL{r_m}{-\al_m}}\right).
 \end{align*}
Result in $(e)$ simply follows from $(d)$ by separating terms that depend on either $\Phi_m$ or $\Phi_s$ and the ones that do not. The 
final step in $(f)$ uses the probability generating functional \cite{ganti2012stochastic} of a PPP and the result of the work in 
\cite{Jeff}, as was used in \eqref{eqn:Eqn10}. Plugging the result of $(f)$ in \eqref{eqn:Eqn11_Pc_P-BS_FD} and substituting the 
expectation with the pdf $f(r_s,r)$ gives the result of \eqref{eqn:Eqn6_Pcus_FD}. 
\section{M-BS coverage probability in IBFD setting}
\label{app:App_MBS_Covg_IBFD}
The coverage probability under M-BS could be derived as shown below:
\begin{equation}
\begin{aligned}
&\Pr\left(\SIR_{um} > T_m, \varepsilon_m\right) \stackrel{(a)}= \E_{r_m',r_s}\left[\Pr(\SIR_{um} > T_m \mid r_m',\,r_s)\right]\\
&= \E_{r_m',r_s}\underbrace{\left[\Pr\left(\frac{ P_m g_{o r_m'} r_m'^{-\alpha_m}}{\sum\limits_{z \in \Phi_m \cap B(o,r_m')^c}P_m g_{oz} \|z\|^{-\alpha_m}+\sum\limits_{z\in\Phi_s\cap B(o,\Delta_s r_m'^{\al_m / \al_s})^c}P_sg_{o z}\|z\|^{-\alpha_s}} > T_m \mid r_m',\,r_s\right)\right]}_{F(\Phi_m, \Phi_s)}.\\
\end{aligned}
\label{eqn:Eqn9_Pcum_FD}
\end{equation}
The result in $(a)$ follows as $r_s$ and $r_m'$ are varied so that event $\varepsilon_m$ is always true which is in accordance with the limits of integration 
in \eqref{eqn:Eqn7}. Let the interference (denominator) term in \eqref{eqn:Eqn9_Pcum_FD} be denoted by $I(\Phi_m,\Phi_s)$, where $\Phi_x$ could be thought of as the process defining the 
entire characteristics of tier $x$. Therefore, $\Phi_m \triangleq \left\{ \lambda_m,P_m,T_m,B_m \right\} \text{ and } \Phi_s~\triangleq~\left\{ \lambda_s,P_s,T_s,B_s \right\}$.
The term $F(\Phi_m, \Phi_s)$ is simplified as follows.
\begin{equation}
\begin{aligned}
F(\Phi_m, \Phi_s) &= \Pr\left(\frac{ P_m g_{o r_m'} r_m'^{-\alpha_m}}{I(\Phi_m,\Phi_s)} > T_m \mid r_m',\,r_s\right)\\
&\stackrel{(a)}= \E_{I(\Phi_m,\Phi_s)}\left[ e^{-(T_m/P_m) r_m'^{\alpha_m} I(\Phi_m,\Phi_s) }  \mid r_m'\right],
\end{aligned}
\end{equation}
where, $(a)$ follows from $g_{o r_m'}$ being a unit mean exponential random variable and $F(\Phi_m,\Phi_s)$ being independent of $r_s$. 
Continuing further,\\
\begin{equation}
\begin{aligned}
&F(\Phi_m, \Phi_s) \stackrel{}= \E_{I(\Phi_m,\Phi_s)}\left[ e^{-T_m r_m'^{\alpha_m} \sum_{z\in \Phi_m \cap B(o,r_m')^c}g_{oz} \|z\|^{-\alpha_m} } e^{-\frac{T_m}{P_m}P_s r_m'^{\alpha_m} \sum_{x\in \Phi_s \cap B(o,\Delta_s r_m'^{\al_m / \al_s})^c}g_{ox} \|x\|^{-\alpha_s} } \mid r_m' \right]\\
&\stackrel{(b)}= \E_{g_{oz},\Phi_m}\left[ \prod_{z\in \Phi_m \cap B(o,r_m')^c}e^{-T_m r_m'^{\alpha_m} g_{oz} \|z\|^{-\alpha_m} } \mid r_m' \right]\, \E_{g_{ox},\Phi_s}\left[ \prod_{x\in \Phi_s \cap B(o,\Delta_s r_m'^{\al_m / \al_s})^c}e^{-\frac{T_m}{P_m}P_s r_m'^{\alpha_m} g_{ox} \|x\|^{-\alpha_s} } \mid r_m' \right]\\
&\stackrel{(c)}= e^{-\pi r_m'^2 \lambda_m T_m^{2/\alpha_m} \int\limits_{T_m^{-2/\alpha_m}}^\infty \frac{1}{1 + t^{\alpha_m/2}}\, \D t}\;e^{ -\pi r_m'^{2 \al_m / \al_s} \lambda_s\left(\frac{P_s T_m}{P_m}\right)^{\frac{2}{\alpha_s}} \int\limits_{\left(\frac{B_s}{B_m T_m}\right)^{\frac{2}{\alpha_s}} }^\infty \frac{1}{1+t^{\alpha_s/2}}\,\D t, },               
\end{aligned}
\label{eqn:Eqn10}
\end{equation}
where $(b)$ follows from the independence of the processes $\Phi_m$ and 
$\Phi_s$ and the assumption of fading on links being independent. Finally,  $(c)$ follows from the single-tier coverage probability result in \cite{Jeff}.

\section{M-BS to User Rate}
\label{app:App_MBS_User_Rate_IBFD}
Rate under M-BS in the IBFD setting could be derived as follows:
\begin{equation}
\begin{aligned}
\E\left[R_{um} \mid \SIR_{um} > T_m\right] &= \E\left[\overline{\eta}\log(1+\SIR_{um})\mid \SIR_{um} > T_m\right]\\
&\stackrel{(a)}= \overline{\eta}\int\limits_{t\ge 0} \Pr(\log(1+\SIR_{um}) > t \mid \SIR_{um} > T_m)\, \D t \\
&= \overline{\eta}\int\limits_{t > 0}  \Pr(\SIR_{um}> 2^t-1 \mid \SIR_{um} > T_m )\, \D t\\
&= \frac{\overline{\eta}}{\Pr(\SIR_{um} > T_m)}\int\limits_{t > 0}  \Pr(\SIR_{um}> \max(2^t-1,\,T_m))\, \D t,
\end{aligned}
\label{eqn:Eqn_Rate_um1}
\end{equation}
where $(a)$ follows from the fact that the rate $R_{um}$ is a positive random variable.

\section{M-BS to P-BS to User Rate ($\min\left(R_{us},\, R_{sm}\right)$)}
\label{app:App_MBS_PBS_User_Rate}
The net rate obtained from P-BS is a minimum over M-BS to P-BS (backhaul) and P-BS to user (access) rates. This is derived as follows:
\begin{equation}
\begin{aligned}
&\E\left[\min\left(R_{us},\, R_{sm}\right)\mid \SIR_{us,sm} > T_{s,b}\right] = \int\limits_{t > 0} \Pr\left( \min\left(R_{us},\, R_{sm}\right) > t \mid \SIR_{us,sm} > T_{s,b} \right) \, \D t\\
&= \int\limits_{t > 0} \Pr\left( R_{us} > t,\, R_{sm} > t \mid \SIR_{us,sm} > T_{s,b} \right)\, \D t\\
&= \int\limits_{t > 0} \Pr\left( \log(1+\SIR_{us}) > t,\,\, \frac{\eta}{n} \log(1+\SIR_{sm}) > t \mid \SIR_{us,sm} > T_{s,b} \right)\, \D t\\
&= \int\limits_{t > 0} \Pr\left( \SIR_{us} > 2^{t}-1,\,\, \SIR_{sm} > 2^{\frac{nt}{\eta}}-1 \mid \SIR_{us,sm} > T_{s,b} \right)\, \D t\\
&= \frac{1}{\Pr(\SIR_{us,sm} > T_{s,b})}\int\limits_{t > 0} \Pr\left( \SIR_{us} > \max(2^{t}-1,\,T_s),\,\, \SIR_{sm} > \max(2^{\frac{nt}{\eta}}-1,\,T_b)\right)\, \D t.
\end{aligned}
\label{eqn:Eqn_RateSc}
\end{equation}

\end{document}